\documentclass[journal]{IEEEtran}
\makeatletter

\let\proof\@undefined
\let\endproof\@undefined
\makeatother
\usepackage{epsfig,subfigure,graphicx}
\usepackage{amsmath}
\usepackage{amsthm}
\usepackage{amssymb}

\newtheorem{thm}{Theorem}
\theoremstyle{definition}
\theoremstyle{remark}
\theoremstyle{plain}
\newtheorem{prop}{Proposition}
\theoremstyle{remark}
\theoremstyle{plain}
\newtheorem{cor}{Corollary}
\theoremstyle{plain}
\newtheorem{lem}{Lemma}

\newcommand{\oper}[1]{\mathcal{#1}}

\newcommand{\norm}[1]{\left\|#1\right\|}
\newcommand{\inner}[1]{\left<#1\right>}
\newcommand{\R}{\mathbb{R}}

\newcommand{\E}{\mathbb{E}}
\newcommand{\Prob}{\mathbb{P}}

\DeclareMathOperator*{\sgn}{sgn}
\DeclareMathOperator*{\Oracle}{Oracle}
\DeclareMathOperator*{\Ber}{Ber}

\DeclareMathOperator*{\TV}{TV}

\begin{document}

\title{Exact recoverability from dense corrupted observations via $L_1$ minimization}

\author{Nam H. Nguyen and Trac D. Tran,~\IEEEmembership{Senior Member,~IEEE}
\thanks{This work has been partially supported by the National Science Foundation (NSF) under Grants CCF-1117545 and CCF-0728893; the Army Research Office (ARO) under Grant 58110-MA-II and Grant
60219-MA; and the Office of Naval Research (ONR) under Grant N102-183-0208.}
\thanks{Nam H. Nguyen and Trac D. Tran are with the Department of Electrical and Computer Engineering, the Johns Hopkins University, Baltimore, MD, 21218 USA (email: {nam, trac}@jhu.edu ).}
\thanks{This work was posted on arxiv.org on Feb. 7, 2011 as 1102.1227.}
}


\maketitle

\begin{abstract}
This paper confirms a surprising phenomenon first observed by Wright \textit{et al.} \cite{WYGSM_Face_2009_J} \cite{WM_denseError_2010_J} under different setting: given $m$ highly corrupted measurements $y = A_{\Omega \bullet} x^{\star} + e^{\star}$, where $A_{\Omega \bullet}$ is a submatrix whose rows are selected uniformly at random from rows of an orthogonal matrix $A$ and $e^{\star}$ is an unknown sparse error vector whose nonzero entries may be unbounded, we show that with high probability $\ell_1$-minimization can recover the sparse signal of interest $x^{\star}$ exactly from only $m = C \mu^2 k (\log n)^2$ where $k$ is the number of nonzero components of $x^{\star}$ and $\mu = n \max_{ij} A_{ij}^2$, even if nearly $100 \%$ of the measurements are corrupted. We further guarantee that stable recovery is possible when measurements are polluted by both gross sparse and small dense errors: $y = A_{\Omega \bullet} x^{\star} + e^{\star}+ \nu$ where $\nu$ is the small dense noise with bounded energy. Numerous simulation results under various settings are also presented to verify
the validity of the theory as well as to illustrate the promising
potential of the proposed framework.
\end{abstract}

\begin{keywords}
Compressed sensing, $\ell_1$-minimization, sparse signal recovery, discrete Fourier transform, (weak) restricted isometry, random matrix, dense error correction.
\end{keywords}

\section{Introduction}

Compressed sensing (CS) has been rigorously studied over a past few years as a revolutionary signal sampling paradigm \cite{Donoho_CS_2006_J}, \cite{CRT_CS_2004_J}, \cite{Rauhut_2004_J}. According to CS, a $k$-sparse signal $x^{\star} \in \R^n$ is measured through a set of linear projections $y_i = \inner{a_i , x^{\star}}$, $i = 1,..., m$, in which vectors $a_i \in \R^n$ form a matrix $A$ of size $m \times n$. The intriguing CS framework advocates the collection of significantly fewer measurements than the ambient dimension of the signal ($m < n$). To reconstruct $x^{\star}$, a standard $\ell_1$-minimization is proposed to solve the inverse problem


\begin{equation}
\label{opt::linear optimization - general approach}
\min_x \norm{x}_1 \quad\quad \text{subject to} \quad\quad y = A x.
\end{equation}

\noindent It has been well known in the literature that if $A$ obeys Restricted Isometry Property (RIP) \cite{CRT_Stability_2006a_J}, \cite{CT_CSdecoding_2005_J} - a property essentially implies that every subset of $k$ or fewer columns of $A$ is approximately an orthogonal system, then the linear program in (\ref{opt::linear optimization - general approach}) is able to faithfully recover $x^{\star}$. This RIP condition has been proven to hold for many types of random measurement matrices \cite{CT_CS_UUP_2006_J}, \cite{RV_CS_2008_J}. For example, random Gaussian or Bernoulli matrices satisfy RIP with high probability as long as the number of measurements $m$ is on the order of $k \log n$ \cite{CT_CS_UUP_2006_J}, whereas the sub-orthogonal matrix $A_{\Omega \bullet}$ sampled uniformly from an orthogonal matrix $A$ obeys RIP with high probability when $m$ is on the order of $k \log^4 n$ \cite{RV_CS_2008_J}.

In many practical applications, we are often interested in situations in which measurements are contaminated by noise. Mathematically, we often observe
$$
y = A x^{\star} + e^{\star},
$$
where $e^{\star} \in \R^m$ is the vector noise. To reconstruct $x^{\star}$ from the observation vector $y$, we minimize the following convex program
\begin{equation}
\label{opt::convex optimization - general approach}
\min_x \norm{x}_1 \quad\quad \text{subject to} \quad\quad \norm{y - A x}_2 \leq \sigma,
\end{equation}
where $\sigma$ is upper bound of the noise level $\norm{e^{\star}}_2$, which assumes to be known. It has been shown in \cite{CRT_Stability_2006a_J}, \cite{DET_2006_J}, \cite{Rauhut_2008_J}, \cite{Tropp_Relax_2006_J} that if $A$ satisfies RIP and $\sigma$ is not too large, then by the same amount of measurements as above, solution $\widehat{x}$ of (\ref{opt::convex optimization - general approach}) does not depart too far from the optimal solution $x^{\star}$. In particular, the authors of \cite{CRT_Stability_2006a_J} proved that the reconstruction error proportionally grows with $\sigma$ as $\norm{\widehat{x} - x^{\star}}_2 \leq C \sigma$, where $C$ is a small numerical constant.

This result is elegant when the noise level is low. However, as the noise energy gets larger, $\widehat{x}$ might be unexpectedly very different from $x^{\star}$. This implies that even a single grossly corrupted measurement may produce $\widehat{x}$ arbitrarily far from the true solution. Unfortunately, gross errors and irrelevant measurements are now ubiquitous in modern applications such as image processing, sensor network, where certain number of measurements may be severely corrupted due to occlusions, sensor failures, transmission error, etc \cite{CLMW_RobustPCA_2009_J}, \cite{WYGSM_Face_2009_J}, \cite{WM_denseError_2010_J}. These examples motivate us to consider a new problem in which we aim to recover a sparse vector $x^{\star}$ from highly corrupted measurements, $y = A x^{\star} + e^{\star}$. In contrast to previous approaches \cite{CRT_Stability_2006a_J}, \cite{DET_2006_J}, \cite{Rauhut_2008_J}, \cite{Tropp_Relax_2006_J} where only small dense noise term $e^{\star}$ is considered, in this paper, entries of $e^{\star}$ can have arbitrarily large magnitude, and their support is assumed to be sparse but unknown. The underlying model has been previously developed by Wright \textit{et al.} \cite{WYGSM_Face_2009_J}. Motivated from the face recognition problem, in which sparse error appears due to a fraction of the query image $y$ being occluded by glasses, hats, etc, the authors proposed to simultaneously minimize the $\ell_1$-norm of both $x$ and $e$,
\begin{equation}
\label{opt::Wright-Ma l1-l1 minimization - sub orthogonal matrix}
\min_{x,e} \norm{x}_1 + \norm{e}_1  \quad\quad \text{subject to} \quad\quad y = A x + e.
\end{equation}
where columns of matrix $A$ are associated with training images. To analyze the model, they assume $A$ obeys the Gaussian distribution \cite{WM_denseError_2010_J}. That is, entries of $A$ are i.i.d $\oper N(0, 1/m)$ Gaussian random variables.

As pointed out by Cand\`es and Romberg \cite{CR_incoherence_2007_J} and Do \textit{et al.} \cite{DGNT_2008_J}, in compressed sensing, completely random measurement matrices might not be relevant in many practical applications. First, we may not be allowed to control measurement matrices. For instance, in MRI or tomography, due to the acquisition system, measurements are inherently frequency based. The second drawback is computationally expensive and memory buffering due to their completely unstructured nature. These weaknesses prevent these fully random sensing matrices from being applied to applications in which both acquisition system (or encoder) and reconstruction system (or decoder) are required to have low complexity and fast implementation.



In this paper, we extend the analysis to a special class of measurement matrices which are constructed from an orthogonal matrix $A$. Let $\Omega$ be a subset of indices of $\{ 1,...,n \}$; and measurement matrix $A_{\Omega \bullet}$ is constructed from rows of $A$ associated with indices in $\Omega$. The observation vector $y$ is now obtained by
\begin{equation}
\label{eqt::measurement model}
y = A_{\Omega \bullet} x^{\star} + e^{\star},
\end{equation}
where we assume that signal $x^{\star}$ and error $e^{\star}$ are sparse vectors whose supports are $T$ and $S$, respectively. These suborthogonal measurement matrices have been carefully studied in the literature such as the partial Fourier ensemble \cite{CR_incoherence_2007_J} and structurally random matrix (SRM) \cite{DGNT_2008_J} as a promising replacement for fully random Gaussian/Bernulli sensing matrices. However, so far, none of the previous work guarantees stable reconstruction under highly corrupted sparse error or a combination of both large sparse error and small dense noise. This is our most significant technical contribution.


To recover $x^{\star}$ and $e^{\star}$, we propose to solve the following extended $\ell_1$-minimization
\begin{equation}
\label{opt::l1-l1 minimization - sub orthogonal matrix}
\min_{x,e} \norm{x}_1 + \lambda \norm{e}_1  \quad\quad \text{subject to} \quad\quad y = A_{\Omega \bullet} x + e,
\end{equation}
where $\lambda > 0$ is a controlled parameter that balance the two $\ell_1$-norm terms.

Surprisingly, with an appropriate choice of $\lambda$, this simple linear program (\ref{opt::l1-l1 minimization - sub orthogonal matrix}) can assure the exact recovery both $x^{\star}$ and $e^{\star}$ exactly,  even when the sparsity of $x^{\star}$ grows almost linearly in the dimension of signal and the errors in $e^{\star}$ are up to a constant fraction of all the entries. This observation will be confirmed via rigorously mathematical justifications as well as extensive simulations in the next few sections.





\subsection{Motivational applications} There are many important applications in which the observations of interest can be modeled as a linear projection of a sparse signal plus sparse error. Before shifting to the presentation of our main results, we briefly introduce several applications and show how well they fit into our underlying model of interest


\begin{itemize}
	
    \item \textit{Image inpainting.}  Given an image $Y$ with missing/corrupted pixels,  we would like to reconstruct the original image by filling in lost information \cite{ESQD_Inpainting_2005_J}. If we assume that errors are indicated by a matrix $E$ whose nonzero-value entries are associated with the missing/corrupted pixels, then $Y$ can be decomposed into two components: the original image $B$ and sparse noise $E$. In image inpainting, the key hypothesis frequently made to guarantee satisfactory performance is that $Y$ has to be sparsely represented by a few coefficients over an overcomplete dictionary $D$ \cite{SED_2005_J}, \cite{ESQD_Inpainting_2005_J}. This dictionary is typically a concatenation of orthogonal transformations, e.g. wavelet, Fourier, DCT or is learned from a set of training images. By denoting $y$, $b$ and $e$ as vectorized versions of matrices $Y$, $B$ and $E$, we have a mathematical representation, $y = D x + e$, where $x$ is the sparse coefficient vector. As opposed to previous works in which locations of missing entries are often required to be known in advance, here we do not need to make any of such assumptions in our model. Rather, utilizing the optimization in (\ref{opt::l1-l1 minimization - sub orthogonal matrix}), we let the algorithm guess both the noisy locations and their magnitudes.

	
    \item \textit{Compressed sensing for networked data.} In sensor networks, the goal is to design a low-power system but still guarantee reliability in transmission. In this setting \cite{HBRN_2008_J}, each sensor collects information of a signal or object $x^{\star}$ by simply projecting $x^{\star}$ onto row vectors $a_i$ of a sensing matrix $A \in \R^{m \times n}$, $b_i = \inner{a_i, x^{\star}}$. As suggested in \cite{HBRN_2008_J}, rather than realizing $A$ in a completely random manner, it is simpler and less computational complex to utilize a matrix $A$ that we can exploit fast implementation and avoid expensive memory buffering such as DCT, Hadarmard or the Fourier transform.


        After having gathered all the data, these sensors send measured information $y_i$ to their neighbors or a central hub for analysis and processing. However, due to the fact that sensors are low cost, it is highly likely that some sensors might fail in collecting data or producing measurements that are not well protected before transmission. This implies that some measurements may be severely corrupted by two types of errors:
        $$
        y = A x^{\star} + e^{\star} + \nu,
        $$
        where $e^{\star}$ is the sparse error, whose entry magnitudes in the support can be arbitrarily large and $\nu$ is dense noise with bounded energy $\sigma$. To recover both $x^{\star}$ and $e^{\star}$, we propose to solve
        $$
        \min_{x, e} \norm{x}_1 + \lambda \norm{e}_1 \quad \text{s.t.} \quad \norm{y - A x - e}_2 \leq \sigma.
        $$

    \item \textit{Joint source-channel coding.} One potential application of CS is simultaneous joint source channel coding \cite{CCZKHBB_2010_C}, \cite{LWW_grossError_2010_C}. In contrast to conventional approach where source data $x^{\star} \in \R^n$ is initially encoded to remove redundancy, then channel-coded for error protection. In CS, $x^{\star}$ is encoded by a simple linear projection $y = Ax$. In \cite{CCZKHBB_2010_C}, to protect the channel, the authors proposed to use more measurements than the optimal value that CS can recover accurately. In order to retrieve $x^{\star}$ under channel error, we need to know the probabilistic model of corrupted entries, which is usually unavailable in practice. We believe that ours is a more accessible and more robust approach in recovering such signal.





\end{itemize}



\subsection{Notations and organization of the paper}

We briefly introduce some notations that will be used throughout the paper. We denote $x_T$ as a vector whose entries are selected from the index set $T \in \{1,..., n\}$ of $x \in\R^n$. Let $\Omega$ be a subset of $\{1,..., n \}$, we denote $A_{\Omega \bullet}$ as a submatrix of $A$, whose rows are taken from $\Omega$. Similarly, $A_{\Omega T}$ denotes a submatrix of $A$, containing rows indexed by $\Omega$ and columns indexed by $T$. Further, we reserve the two index sets $T$ and $S$ for signal ($x^{\star}$) and error ($e^{\star}$) supports. The sparsity level of $x^{\star}$ and $e^{\star}$ are $k= |T|$ and $s = |S|$, respectively. For a vector $x$, $\sgn(x)$ represents the sign of $x$ componentwise.

We will use several standard vector and matrix norms, which we simply present here for completeness. For $x \in \R^n$, $\norm{x}_1 = \sum_{i=1}^n |x_i|$ is the $\ell_1$-norm, $\norm{x}_2$ is the $\ell_2$-norm and $\norm{x}_{\infty} = \max_i |x_i|$ is the $l_{\infty}$-norm. For matrices $B$, we only use the spectral norm, denoted by $\norm{B}$.

We denote by $C, C_1, c, c_1, ...$ positive absolute constants. Finally, when we say that an event occurs with high probability, we mean the occurring probability of the event is at least $1 - c n^{-1}$.

The remainder of this paper is structured as follows. The main results are introduced in Section \ref{sec::main results}. Our proof structure is described Section \ref{sec::proof structure}. Supporting results are subsequently presented in Sections \ref{sec::dual certificate} and \ref{sec::proof of last 2 theorems}. Section \ref{sec::oracle inequalities} compares our results with the oracle in which we know in advance the locations of signal and error support. We demonstrate the consistency of our results via extensive simulation in Section \ref{sec::experiments}. Finally, Section \ref{sec::conclusion} summarizes this paper and makes some closing remarks.

\section{Main results}
\label{sec::main results}

\subsection{Sparse model}


We begin by studying the easier problem where signal $x^{\star}$ is perfectly k-sparse and observation vector $y$ is also corrupted by sparse error. A more difficult problem with non-sparse signal $x^{\star}$ and $y$ being corrupted by both sparse and dense noise will be subsequently investigated in this section. Toward the end, we denote the sparsity indices of $x^{\star}$ and $e^{\star}$ as $k$ and $s$ and introduce the $(k,s)$\textit{-sparse model} defined as follows:
\begin{itemize}
	\item Signs of $x^{\star}$ at the support $T$ is independently and equally likely to be $1$ or $-1$.
	\item Support $S$ of $e^{\star}$ is uniformly distributed among all sets of cardinality $s$ in $\Omega$.
\end{itemize}

The random assumption on the sign of $x^{\star}$ at support $T$ is typical in compressed sensing \cite{CR_incoherence_2007_J}. This assumption is a sufficient rather than necessary condition and is employed for the convenience of our proof only. Indeed, by sacrificing a factor of $\log n$ to the number of measurements, we can establish similar results when the signs of $x^{\star}$ arbitrary. We refer the interested readers to a recent paper \cite{CP_RIPless_2010_J} for more details.



\subsection{Exact recovery as measurements are corrupted by sparse noise}


\begin{thm}
\label{thm::main theorem}
Let $x^{\star}$ be a fixed vector in $\R^n$ and $A$ be an $n \times n$ orthonormal matrix ($A^* A = I$) with $|A_{ij}|^2 \leq \frac{\mu}{n}$, where $1 \leq \mu \leq n$, and assume that $(x^{\star}, e^{\star})$ is taken from the $(k,s)$-sparse model. Suppose we observe $m$ entries from the projection $A x^{\star}$ with locations in $\Omega$ sampled uniformly at random and these entries are then corrupted by noise $e^{\star}$. Then there exist numerical constants $c$ and $C$ such that with probability at least $1 - c n^{-1}$, the convex program (\ref{opt::l1-l1 minimization - sub orthogonal matrix}) with $\lambda \sim \sqrt{\frac{n}{\mu m \log n}}$ correctly recovers both the signal and the error (i. e. $\widehat{x} = x^{\star}$ and $\widehat{e} = e^{\star}$), provided that
\begin{equation}
\label{thm::conditions of the main theorem}
m \geq C \mu^2 k (\log n)^2 \quad\quad \text{and}\quad s \leq \gamma m,
\end{equation}
for any $\gamma$ close to $0.9$.
\end{thm}

In other words, Theorem \ref{thm::main theorem} asserts a surprising message: a sparse signal $x^{\star}$ can be faithfully recovered with probability converging to one from arbitrary and completely unknown corrupted patterns (as long as they are randomly distributed). We do not place any assumption on the magnitudes or signs of the nonzero entries of $e^{\star}$. In fact, its magnitude can be arbitrarily large. Theorem \ref{thm::main theorem} is generic in the sense that it only requires signs of nonzero entries of $x^{\star}$ to be uniformly distributed; everything else is deterministic. We believe that the random assumption on the sign pattern is artificial and can be removed. Indeed, when $A$ is a Fourier matrix, applying advanced techniques in \cite{CRT_CS_2004_J}, we are able to obtain Theorem \ref{thm::main theorem} for all $x^{\star}$ supported on $T$. An interesting open problem is whether this result also holds for other orthogonal sensing matrices.




It is necessary to further clarify Theorem \ref{thm::main theorem}. First, higher probabilities of success (i.e. in the form $1 - c n^{-\beta}$ with $\beta \geq 1$) can be obtained at the expense of increasing the number of observations by a factor of $\beta$. Next, the theorem addresses that for a particular selection of $\Omega$, exact recovery only holds for an arbitrary \textit{fixed} sparse signal with high probability (as long as signs of such signal at its support are uniformly distributed). In other words, there is no uniform sparse signal recovery guaranteed here. In fact, in order to establish perfect recovery for all sparse signal, we might have to require certain stronger properties for matrix $A_{\Omega \bullet}$ such as RIP \cite{CRT_Stability_2006a_J} or similar to RIP. As shown in \cite{CT_CS_UUP_2006_J}, \cite{RV_CS_2008_J}, $A_{\Omega \bullet}$ obeys RIP with high probability only if the number of measurements exceeds $C k \log^4 n$, which is a far inferior requirement than our optimal value. By relaxing RIP, we are able to significantly reduce the amount of measurements needed and are still able to guarantee perfect recovery even when the data is highly corrupted.

Theorem \ref{thm::main theorem} also implies that up to a $\log n$ factor from the optimal number of observations as in compressed sensing, we are able to precisely recover the signal in the presence of gross error. In the following theorem, we establish that by the same order of $k \log n$ measurements as in compressed sensing, $\ell_1$-minimization is still able to recover precisely both spare signal and high-energy sparse noise. In particular, we draw an interesting relationship between signal sparsity, error sparsity and the parameter $\lambda$.

\begin{thm}
\label{thm::detail of the main theorem}
Let $x^{\star}$ be a fixed vector in $\R^n$ and $A$ be an $n \times n$ orthonormal matrix ($A^* A = I$) with $|A_{ij}|^2 \leq \frac{\mu}{n}$, where $1 \leq \mu \leq n$ and assume that $(x^{\star}, e^{\star})$ is taken from the $(k,s)$-sparse model. Suppose that we observe $m$ entries from the projection $A x^{\star}$ with locations in $\Omega$ sampled uniformly at random and these entries are then further corrupted by noise $e^{\star}$. Then there exist numerical constants $c$, $C_1$ and $C_2$ such that with probability at least $1 - c n^{-1}$, the convex program in (\ref{opt::l1-l1 minimization - sub orthogonal matrix}) with $\lambda = \sqrt{\frac{1}{\gamma \log n}\frac{n}{\mu m}}$, $\gamma \in (0,1)$ correctly recovers both the signal and the error (i. e. $\widehat{x} = x^{\star}$ and $\widehat{e} = e^{\star}$), provided that
\begin{equation}
\label{inq::conditions on m and s of the main theorem}
m \geq C_1 \mu^2 \max \{ \frac{\gamma}{(1-\gamma)^2} k  (\log n)^2, k \log n , (\log n)^2 \}
\end{equation}
\begin{equation}
\text{and} \quad s \leq C_2 \gamma m.
\end{equation}
\end{thm}

It is easy to check that Theorem \ref{thm::detail of the main theorem} implies Theorem \ref{thm::main theorem} by setting $\gamma = 0.9$, or equivalently $\lambda = \sqrt{n/0.9 \mu m \log n}$. Later in the paper, we focus on establishing this theorem, then Theorem \ref{thm::main theorem} will automatically follow.


We would like to note the significance of the parameter $\mu$ here: $\mu$ can be seen as the incoherence of the matrix $A$, which measure how concentrated or expanded rows of measurement matrix $A_{\Omega\bullet}$ are. Since $A$ is orthonormal, the value of $\mu$ ranges between $1$ and $n$. In the worse case scenario when rows of $A$ are maximally concentrated, then $\mu = n$ and $A$ is the identity matrix. It is clear in this case that we cannot retrieve $x^{\star}$ under a single gross error even if all $n$ measurements (which is now the signal $x^{\star}$ itself) are observed. On the other hand when $\mu  = 1$, entries of $A$ are perfectly spread out and the number of measurements attains its optimally minimum value.

It can be seen that $\lambda$ in (\ref{opt::l1-l1 minimization - sub orthogonal matrix}) controls the balance between two terms: $\norm{x}_1$ and $\norm{e}_1$. Specifically, if a large value of $\lambda$ is selected, we expect to recover the denser-support signal but under sparser error. On the other hand, a smaller choice of $\lambda$ is better when the error is denser while the signal is sufficiently sparse. Theorem \ref{thm::detail of the main theorem} mathematically indicates that it is actually the case. In particular, if $\gamma$ is chosen to be $1/\log n$, then relying on only $m = C k \log n$ measurements, linear (convex) programming (\ref{opt::l1-l1 minimization - sub orthogonal matrix}) not only recovers the $k$-sparse $x^{\star}$ faithfully, it is also able to correctly identify the noise with arbitrary large magnitude as long as the noise sparsity is proportional to $m /\log n$. On the contrary, if we set $\gamma$ close to one, then (\ref{opt::l1-l1 minimization - sub orthogonal matrix}) can retrieve $x^{\star}$ whose support is $m/C k (\log n)^2$ under error whose support is up to a  constant fraction of all the measurements. In fact, the theorem gives a whole range of $\lambda$ values, whose selection might rely on the prior information we can collect about the sparsity level of the signal as well as of the noise.

\subsection{Stable recovery as measurements are corrupted by both dense and sparse errors}

Our result in Theorem \ref{thm::detail of the main theorem}, although interesting, is limited to the case of noise being exactly sparse only. In practical applications, observations $y$ are also often contaminated by dense noise, which can be either deterministic or stochastic. In this section, we investigate the model where observations are corrupted by both the unknown dense noise $\nu$ with small energy bound $\norm{\nu}_2 \leq \sigma$ and the sparse noise $e^{\star}$, whose magnitudes of nonzero entries are arbitrarily large
$$
y = A_{\Omega \bullet} x^{\star} + e^{\star} + \nu.
$$

At first, for the ease of demonstrating our results as well as proving technique, we consider a particular situation where the observation $y$ is only corrupted by dense error whose energy is bounded by $\sigma$. The problem is now to recovery $x^{\star}$ from noisy observation $y$, where
$$
y = A_{\Omega \bullet} x^{\star} + \nu.
$$

To recover $x^{\star}$, it has been well established that we need to minimize the following convex program
\begin{equation}
\label{opt::l1-l1 minimization with noise and without sparse noise - sub orthogonal matrix}
\min_{x} \norm{x}_1   \quad\quad \text{subject to} \quad\quad \norm{b - A_{\Omega \bullet} x }_2 \leq \sigma.
\end{equation}


When the observation vector $y$ is clean, Cand\`es and Romberg \cite{CR_incoherence_2007_J} showed that the $\ell_1$-minimization is able to recover $x^{\star}$ precisely. In this section, we extend their result and prove that even with imperfect observations $y$, the convex program is stable vis a vis perturbations. Particularly, the recovery error is bounded away by a factor of $\sigma$. To the best of our knowledge, this is the first robust recovery bound when measurements taken from suborthogonal matrices are corrupted by deterministic noise.

\begin{thm}
\label{thm::main theorem - with dense noise and without sparse noise}
Under the same assumptions defined in Theorem \ref{thm::detail of the main theorem} and provided that there exists a numerical constant $C$ such that $m \geq C \mu k \log n$, for any perturbation $\nu$ with $\norm{\nu}_2 \leq \sigma$, the solution $\widehat{x}$ to the convex program in (\ref{opt::l1-l1 minimization with noise and without sparse noise - sub orthogonal matrix}) yields
\begin{equation}
\norm{\widehat{x} - x^{\star}}_2 \leq 8 \sigma \sqrt{n (1 + 2n/m)} + 2 \sigma.
\end{equation}
\end{thm}

Roughly speaking, Theorem \ref{thm::main theorem - with dense noise and without sparse noise} states that for a family of matrices $A_{\Omega}$ constructed from any unitary matrix $A$, minimizing the $\ell_1$-norm stably recovers $\widehat{x}$ from just $O(\mu k \log n)$ measurements. A direct consequence of this theorem says that as $\sigma$ comes closer to zero, the solution of (\ref{opt::l1-l1 minimization with noise and without sparse noise - sub orthogonal matrix}) is exact, which coincides with Cand\`es and Romberg's result \cite{CR_incoherence_2007_J}. Moreover, our result is established for any deterministic noise $\nu$. While preparing this manuscript, we learned of an independent investigation of Cand\`es and Plan \cite{CP_RIPless_2010_J} into this problem. They place stochastic assumptions on $\nu$, e.g. $\nu$ obeys the Gaussian distribution, and thus the resulting error bound is improved.


A more challenging situation occurs when observations are not only contaminated by dense noise with small energy, but they are also corrupted by sparse noise with arbitrarily large magnitude. This model includes the previous settings in Theorems \ref{thm::detail of the main theorem} and \ref{thm::main theorem - with dense noise and without sparse noise} as the particular cases:
\begin{equation}
\label{eqt::observation b under sparse signal, sparse noise and dense noise}
y = A_{\Omega \bullet} x^{\star} + e^{\star} + \nu.
\end{equation}

\noindent To successfully recover $x^{\star}$ (as well as $e^{\star}$), we propose to minimize the following convex program
\begin{equation}
\label{opt::l1-l1 minimization with noise - sub orthogonal matrix}
\min_{x,e} \norm{x}_1 + \lambda \norm{e}_1  \quad\quad \text{s. t.} \quad\quad \norm{b - A_{\Omega \bullet} x - e}_2 \leq \sigma
\end{equation}
where $\sigma$ is the bound of energy noise $\nu$, assumed to be known.

\begin{thm}
\label{thm::main theorem - with dense noise and sparse noise}
Under the same assumptions defined in Theorem \ref{thm::detail of the main theorem} and provided that there exist numerical constants $C_1$ and $C_2$ such that
\begin{equation}
\label{inq::conditions on m and s of the main theorem}
m \geq C_1 \mu^2 \max \{ \frac{\gamma}{(1-\gamma)^2} k  (\log n)^2, k \log n , (\log n)^2 \}
\end{equation}
\begin{equation}
\text{and} \quad s \leq C_2 \gamma m,
\end{equation}

\noindent then the pair of solution $(\widehat{x}, \widehat{e})$ to the convex program (\ref{opt::l1-l1 minimization with noise - sub orthogonal matrix}) obeys
\begin{equation}
\label{inq::bound l-2 norm of widehat-x - x* plus widehat-e - e*}
\norm{\widehat{x} - x^{\star}}_2 + \norm{\widehat{e} - e^{\star}}_2 \leq \frac{8 (\lambda + 1) \sigma}{\min \{ 1, \lambda \}} \sqrt{n \left( 1 + \frac{4 n}{m-s}\right) }  +  2\sigma.
\end{equation}
\end{thm}

Theorem \ref{thm::main theorem - with dense noise and sparse noise} is significant because it addresses that the convex program in (\ref{opt::l1-l1 minimization with noise - sub orthogonal matrix}) can reliably reconstruct the sparse signal even when the measurements are severely corrupted by both gross sparse and small dense errors. This is the situation that we most likely will encounter in practical applications. When the measurement $y$ is not corrupted by the dense noise $\nu$, the signal can be reconstructed perfectly, regardless of how large the sparse noise is, as previously mentioned in Theorem \ref{thm::detail of the main theorem}. In addition, we will demonstrate in Section 5 that this reconstruction error is optimal up to a $\sqrt{n}$ factor as compared to the oracle situation in which locations of $T$ nonzero entries of signal $x^{\star}$ as well as $S$ nonzero entries of the sparse error $e^{\star}$ are known in prior. In particular, if ignoring this $\sqrt{n}$ factor, (\ref{inq::bound l-2 norm of widehat-x - x* plus widehat-e - e*}) is unimprovable.


The preceding results have focused on scenarios where the signal is perfectly sparse. We now consider probably the most general setting, in which $x^{\star}$ is not exactly sparse, but rather can be approximated well by a sparse vector and the observation vector $y$ is corrupted by both sparse error and dense noise with noise level $\sigma$. Denote $x^{\star}_T \in \R^n$ as a vector containing the $k$ largest magnitude entries of $x^{\star}$ and zeros elsewhere and assume an uniform distribution on the sign of $x^{\star}_T$ at the support $T$, we can now establish the following corollary

\begin{cor}
\label{corr::main corr - with generic signal and dense noise and sparse noise}
Under the same assumptions defined in Theorem \ref{thm::detail of the main theorem}, the pair of solution $(\widehat{x}, \widehat{e})$ to the convex program (\ref{opt::l1-l1 minimization with noise - sub orthogonal matrix}) obeys
\begin{equation}
\label{corr::with dense noise and sparse noise and generic signal}
\begin{split}
&\norm{\widehat{x} - x^{\star}}_2 + \norm{\widehat{e} - e^{\star}}_2 \leq  \frac{8 (\lambda + 1) }{\min \{ 1, \lambda \}} \\
&{ }\times \left[ \sigma \sqrt{n \left( \frac{4 n}{m-s} + 1 \right) } + 2\norm{x^{\star} - x^{\star}_{T} }_1 \right]  + 2 \sigma.
\end{split}
\end{equation}
\end{cor}

Ignoring the $\sigma \sqrt{8}$ term, one can see how the bound in Corollary \ref{corr::main corr - with generic signal and dense noise and sparse noise} shows a natural splitting into two terms. The first can be interpreted as \textit{data error} associated with the noise $\nu$, whereas the second term relates to the \textit{approximation error}, measuring how far the signal $x^{\star}$ is from the best $k$-sparse approximation $x^{\star}_T$.

\subsection{When error is sparsified under an arbitrary basis}

Thus far, we have only investigated truly sparse error $e^{\star}$. That is, $e^{\star}$ is sparse under the identity transformation. A natural generation is to consider $e^{\star}$ being sparse under any orthogonal transformation $D$, including the former as a particular case. Mathematically, we consider the observation model
\begin{equation}
\label{eqt::observation b under sparse signal, sparse noise under D and dense noise}
y = A_{\Omega \bullet} x^{\star} + D g^{\star} + \nu,
\end{equation}
where $e^{\star} = D g^{\star}$ and $g^{\star}$ is a $s$-sparse vector. It is clear that via simple algebra, this setting boils down to (\ref{eqt::observation b under sparse signal, sparse noise and dense noise}) as
$$
D^* y = D^* A_{\Omega \bullet} x^{\star} + g^{\star} + D^* \nu.
$$

\noindent Notice that due to the orthogonality of $D$, $D^* A_{\Omega \bullet}$ is also an orthogonal matrix. Therefore, all preceding theorems are still relevant in this setting. The parameter $\mu$ is now interpreted as the mutual incoherence between the sensing matrix $A_{\Omega}$ and the sparsifying transform $D$. In particular,
\begin{equation}
\mu = n \max_{i,j} | \inner{a_i, d_j} |,
\end{equation}
where $a_i$ and $d_j$ are columns of matrices $A_{\Omega \bullet}$ and $D$. As the incoherence of these two matrices is small, fewer measurements are required to still guarantee stable recovery. This results from an intuitive fact that it is easier to decompose $y$ into $x^{\star}$ and $g^{\star}$ if two column spaces of $A_{\Omega \bullet}$ and $D$ are sufficiently separated.


\subsection{Contribution and connections to previous works}


The problem of recovering the signal from grossly corrupted measurements has initially been formulated by Wright \textit{et al.} in an appealing practical paper \cite{WMMSHY_2010_J} and further analyzed in \cite{WM_denseError_2010_J}.  Taking the sparsity information of $e^{\star}$ into account, the authors proposed to solve
\begin{equation}
\label{opt::l1-l1 minimization - Gaussian matrix}
\min_{x, e} \norm{x}_1 + \norm{e}_1 \quad\quad \text{subject to} \quad\quad y = A x + e.
\end{equation}

\noindent The result of \cite{WM_denseError_2010_J} is asymptotic in nature. The authors showed that as $n$ is extremely large and provided $x$ is extremely sparse, then (\ref{opt::l1-l1 minimization - Gaussian matrix}) can precisely recover both $x^{\star}$ and $e^{\star}$ from almost any error with support fraction bounded away from $100 \%$. Their analysis is based on the Gaussian assumption of the matrix $A$. Particularly, $A$ is a matrix whose columns $a_i$'s are assumed to be $\oper N(\mu, \frac{\nu^2}{m} I_m)$, where $\norm{\mu}_2 = 1$ and $\norm{\mu}_{\infty} \leq C m^{-1/2}$. Furthermore, for sufficiently large $m$, they require the sparsity of $x$ to grow sublinearly with $m$. This is of course far from the optimal bound, in which $k$ is almost linear with $m$ (i.e. only in the order of $m/\log n$).

One of the appealing consequence of their analysis is an explicit expression between three important terms: the dimension ratio $\delta = \frac{n}{m}$ of the matrix $A \in \R^{m \times n}$, the fraction error $\rho = \frac{s}{m}$ and the signal support density $\alpha = \frac{k}{m}$. However, this relationship is difficult to interpret due to the complicated coupling of these terms.


Employing the idea from \cite{WM_denseError_2010_J}, Li \textit{et al.} \cite{LWW_grossError_2010_C} and Laska \textit{et al.} \cite{LDB_corruption_2009_C} proposed different applications under the same framework. The former considered the problem of joint source-channel coding, and the later proposed a so-called pursuit of justice model to deal with sparse unbounded noise. When the measurement matrix $A$ obeys restricted isometry property (RIP), both of them showed that the combination matrix $[A, \text{ }I]$ also satisfies the RIP with high probability, where $I$ is the identity matrix. A consequent conclusion is that the signal is perfectly recovered as long as signal and error sparsity levels are in the order of $m/\log n$. The main drawback of these papers is that they are not able to show that perfect recovery is guaranteed when the number of corrupted entries is linearly proportional to the total number of observations.

After the initial submission of our paper to Arxiv, we noticed another two independent investigations into this problem: Studer \textit{et al.} \cite{SKPB_2011_C} and Li \cite{Li_2011_J}. The former studies the more general observation model, $y = A x + D e$, where $A$ and $D$ are general matrices. The authors established deterministic guarantee, which is weaker than our results in Theorem \ref{thm::main theorem}. Using different proof techniques, the latter paper delivered similar results as in Theorem \ref{thm::main theorem} with more general model of the sensing matrix $A_{\Omega \bullet}$. In particular, rows of $A_{\Omega \bullet}$ are sampled independently from a population $F$ obeying $\E a_i a_i^* = I$. However, both papers do not investigate the more realistic model in which both sparse and dense noise present in the observations.


In another direction and much earlier, Cand\`es and Tao investigated the error correction problem \cite{CT_CSdecoding_2005_J}. In this problem, the question is how to reconstruct the input vector $x^{\star}$ from corrupted measurements $y = B x^{\star} + e^{\star}$, where the coding matrix $B \in \R^{m \times n}$ is required to be overcomplete ($m > n$) and $e^{\star}$ is the channel corruption vector, which is usually assumed to be sparse. They proposed to retrieve $x^{\star}$ by solving the following $\ell_1$-minimization problem
\begin{equation}
\label{opt::Candes optimization}
\min_x \norm{b - B x}_1 .
\end{equation}

Though sharing the same general $\ell_1$ model, our approach departs from all previous work in compressed sensing in many aspects:


1) Unlike Wright and Ma's model \cite{WM_denseError_2010_J} where Gaussian measurement matrices are analyzed, we study the problem with suborthogonal matrices. These matrices often possess many desirable properties over Gaussian matrices in term of fast and efficient computation \cite{CR_incoherence_2007_J}, \cite{DGNT_2008_J}. Furthermore, we investigate the more difficult problem in which both sparse and dense error appear in the observations. This model is not studied in \cite{WM_denseError_2010_J}. We show a surprising message that the extended $\ell_1$ minimization is stable under both perturbations, even if the sparse error is arbitrarily large and its support size is arbitrarily close to the total number of observations. A straight forward consequence of this result is that accurate recovery is achieved when measurements are not perturbed by dense noise.



2) Our model is different from Cand\`es and Tao \cite{CT_CSdecoding_2005_J} in two aspects. First, we allow the coding matrix to be under-determined, that is $m \leq n$. Second, the input vector is assumed to be sparse. If we recast the extended $\ell_1$-minimization in (\ref{opt::l1-l1 minimization - sub orthogonal matrix}) as
\begin{equation}
\min_x \norm{x}_1 + \lambda \norm{b - A_{\Omega \bullet} x}_1,
\end{equation}
then one can clearly see the integration of the two $\ell_1$-norms in a unified optimization: one is used to impose sparsity of the input vector whereas the other exploits error sparsity as in (\ref{opt::Candes optimization}).

3) We propose a minor but subtle modification in the extended $\ell_1$ minimization of \cite{WM_denseError_2010_J}. By adding a regularization parameter $\lambda$ into (\ref{opt::l1-l1 minimization - sub orthogonal matrix}), we can balance the $\ell_1$-norm of both $x$ and $e$. Specifically, we can establish an explicit expression for the regularization parameter $\lambda$ as well as the sparsity levels of both signal and error. This mathematical expression is intuitively interpretable: signal and error sparsity levels should be inversely proportional. If more measurements are corrupted $-$ equivalently, the error is denser $-$ we expect to recover the signal with smaller support size. In contrast, we are able to recover the signal with larger support size when fewer errors appear in the measurement vector. In practice, when the fraction of error is unknown, we can set a good-for-all parameter $\lambda = \sqrt{\frac{n}{m \log^{1/2} n}}$.


\section{Structure of our proof}
\label{sec::proof structure}

\subsection{Bernoulli model and derandomization technique}

\textbf{The Bernoulli model. } Instead of showing that Theorem \ref{thm::detail of the main theorem} holds as $\Omega$ and $S$ are sets of size $m$ and $s$ sampled uniformly at random, we find that it is more convenient to prove the theorem for subsets $\Omega$ and $S$ sampled according to the Bernoulli model. This way, we can take advantage of the statistical independence of measurements. The same argument as presented in \cite{CRT_CS_2004_J}, \cite{CLMW_RobustPCA_2009_J} shows that the probability of "failure" under the uniform model is less than two times the probability of failure under the Bernoulli model. Here, "failure" implies the optimization in (\ref{opt::l1-l1 minimization - sub orthogonal matrix}) does not recover exactly the signal. Thus, from now on, we instead consider $\Omega = \{ i \in [1,n]: \delta_i = 1 \}$ where $\{\delta_i\}_{1 \leq i \leq n}$ is a sequence of independent identically distributed Bernoulli random variables taking value one with probability $\eta$ and zero with probability $1 - \eta$, where $\eta$ is chosen such that the expected cardinality of $\Omega$ is $\eta n = m$. Similarly, let $S = \{ i \in \Omega: \delta'_i = 1 \}$, $1 \leq i \leq m$ where $\{\delta'_i \}_{i \in \Omega}$ are i.i.d Bernoulli random variables with $\Prob(\delta'_i = 1) = \rho$ so that the expected cardinality of $S$ is $\rho m = s$. Toward this end, we will write $\Lambda \sim \Ber(\eta)$ as a shorthand for $\Lambda$ sampled from the Bernoulli model with parameter $\eta$.

The following are five important index sets that is frequently used in the sequel.
\begin{itemize}
    \item $\Omega$ are those locations corresponding to observations: $\Omega \sim \Ber (\eta)$ with $\eta = \frac{m}{n}$.
    \item $S \subset \Omega$ are locations where the measurements are available but absolutely unreliable. It is clear that the distribution of $S$ relies on that of $\Omega$. Conditioning on $\Omega$, we have $S \sim \Ber(\rho)$ with $\rho = \frac{s}{m}$. We can also think $S$ as a subset selected from the set $ \{1,..., n\}$ with parameter $\eta \rho$. That is, $S \sim \Ber( \eta \rho )$.
    \item $J \subset \Omega$ are locations where the measurements are available and truthworthy. It is clear that $J = \Omega/S$. Conditioned on $\Omega$, we have $J \sim \Ber(1-\rho)$. In other words, $J \sim \Ber ( \rho_0)$ with $\rho_0 := \eta (1 - \rho)$.
    \item We also consider the index sets $S^c = \{1,...,n \} / S$ and $J^c = \{1,...,n \} / J$
\end{itemize}

\textbf{Derandomization. } In Theorem \ref{thm::detail of the main theorem}, the sign of $e^{\star}$ is fixed. During the proof, we need to place an additional assumption on $e^{\star}$. That is, the sign of $e^{\star}_S$ is uniformly distributed, receiving value $1$ or $-1$ with probability $1/2$. However, by the same appealing derandomization technique presented in \cite{CLMW_RobustPCA_2009_J}, the probability of recovering $e^{\star}$ whose signs on the support $S$ are arbitrary is at least equal to that of recovering $e^{\star}$ whose signs are equally likely to be $1$ or $-1$. This is formally stated in the lemma below

\begin{lem} [\textit{Theorem 2.3 of \cite{CLMW_RobustPCA_2009_J}}]
Suppose $x^{\star}$ obeys conditions of Theorem \ref{thm::detail of the main theorem} and the locations of nonzero entries of $e^{\star}$ follows the Bernoulli model with parameter $2 \rho$, and signs of $e^{\star}$ are i.i.d $\pm 1$ with probability $1/2$. Then, if the solution of extended $\ell_1$-minimization (\ref{opt::l1-l1 minimization - sub orthogonal matrix}) is exact with high probability, it is also exact with at least the same probability with the model in which the signs of $e^{\star}$ are fixed and its nonzero entries are selected from the Bernoulli model with parameter $\rho$.

\end{lem}





\subsection{Dual certificate}

The following lemma shows that if there exists a dual pair ($z^{(x)}, z^{(e)}$) satisfying certain conditions, then for any pair $(x, e)$, its $\ell_1$-norm sum is no smaller than that of $(x^{\star}, e^{\star})$ .

\begin{lem}
\label{lem::construction of dual certificate (z, e) - Sub orthogonal matrix - tradition}
Suppose that $\norm{A_{J^c T}} < 1$. If there exists a pair of dual vectors ($z^{(x)}, z^{(e)}$) with the following properties,
\begin{enumerate}
    \item $z^{(x)} = \lambda A^*_{\Omega \bullet } z^{(e)}$,
    \item $z^{(x)}_T = \sgn(x^{\star}_T)$ and $\norm{z^{(x)}_{T^c}}_{\infty} \leq 3/4$,
    \item $z^{(e)}_S = \sgn(e^{\star}_S)$ and $\norm{z^{(e)}_{S^c}}_{\infty} \leq 3/4$,
\end{enumerate}
then for any perturbation pair ($h, f$) satisfying $f = - A_{\Omega \bullet} h$, we have
\begin{equation}
\label{inq::lemma construction of dual certificate - sub orthogonal matrix - inequality}
\begin{split}
&\norm{x^{\star} + h}_1 + \lambda \norm{e^{\star}+f}_1 \\
& \geq \norm{x^{\star}}_1 + \lambda \norm{e^{\star}}_1 + \frac{1}{4} (\norm{h_{T^c}}_1 + \lambda \norm{A_{J \bullet} h}_1).
\end{split}
\end{equation}
\end{lem}

Before proving this lemma, it is necessary to notice how the Lemma implies the perfect recovery of the linear program in (\ref{opt::l1-l1 minimization - sub orthogonal matrix}). Indeed, denote by $(\widehat{x}, \widehat{e})$ the optimal solution of (\ref{opt::l1-l1 minimization - sub orthogonal matrix}) and let $\widehat{x} := x^{\star} + h$ and $\widehat{e} := e^{\star} + f$, then it is obvious that $f = - A_{\Omega \bullet} h$. By the convexity of the objective function, we have $\norm{x^{\star} + h}_1 + \lambda \norm{e^{\star}+f}_1 \leq \norm{x^{\star}}_1 + \lambda \norm{e^{\star}}_1$.

Furthermore, from Lemma \ref{lem::construction of dual certificate (z, e) - Sub orthogonal matrix - tradition}, assuming the existence of a dual pair $(z^{(x)}, z^{(e)})$ and $\norm{A_{J^c T}} < 1$, the inequality (\ref{inq::lemma construction of dual certificate - sub orthogonal matrix - inequality}) obeys. Combining both arguments, we have
$$
\frac{1}{4} (\norm{h_{T^c}}_1 + \lambda \norm{A_{J \bullet} h}_1) \leq 0.
$$

\noindent It is clear that the left-hand side of the above equation is strictly greater than $0$ for every $h \neq 0$. Thus, in order for the equality to occur, it is necessary that $h_{T^c} = 0$ and $A_{J T} h_T = 0$. We can establish that, due to the orthogonality of matrix $A$, the condition $\norm{A_{J^c T}} < 1$ is equivalent to $\norm{I - A^*_{J T} A_{J T}} < 1$. This suggests that $A^*_{JT} A_{J T}$ is invertible, and thus, $A_{J T} h_T = 0$ only if $h_T = 0$. We therefore conclude that $h=0$ and $f=- A_{\Omega \bullet} h=0$ or in other words, $(\widehat{x}, \widehat{e})$ is the exact solution.

\begin{proof} [Proof of Lemma \ref{lem::construction of dual certificate (z, e) - Sub orthogonal matrix - tradition}]

\noindent Denote as $v_0$ and $w_0$ the subgradients of $\norm{x}_1$ and $\norm{e}_1$ at $x^{\star}$ and $e^{\star}$, respectively. It is well-known that $v_{0_T} = \sgn(x^{\star}_T)$ and $\norm{v_{0_{T^c}}}_{\infty} \leq 1$. Similarly, we have $w_{0_S} = \sgn(e^{\star}_S)$ and $\norm{w_{0_{S^c}}}_{\infty} \leq 1$. By the definition of subgradients, we derive
\begin{equation}
\begin{split}
\label{inq::bound ||x*+h||_1 + lambda ||e*+f||_1 }
&\norm{x^{\star} + h}_1 + \lambda \norm{e^{\star} + f}_1 \\
&\geq  \norm{x^{\star}}_1 + \lambda \norm{e^{\star}}_1 + \inner{v_0, h} + \lambda \inner{w_0, f} \\
&= \norm{x^{\star}}_1 + \lambda \norm{e^{\star}}_1 + \inner{v_0, h} - \lambda \inner{w_0, A_{\Omega \bullet} h}.
\end{split}
\end{equation}

 Let us now consider $\inner{v_0, h} - \lambda \inner{w_0, A_{\Omega \bullet} h}$. By decomposing $v_0$ and $w_0$ into vectors of index sets $\{S,T \}$ and their complements $\{S^c, T^c \}$, we have
\begin{align*}
\inner{v_0, h} - \lambda \inner{w_0, A_{\Omega \bullet} h} &= \inner{\sgn(x^{\star}_T), h_T} - \lambda \inner{\sgn(e^{\star}_S), A_{S \bullet} h} \\
&{ } + \inner{v_{0_{T^c}}, h_{T^c}} - \lambda \inner{w_{0_J}, A_{J \bullet} h}.
\end{align*}

\noindent Now choosing $v_{0_{T^c}}$ such that $\inner{v_{0_{T^c}}, h_{T^c}} = \norm{h_{T^c}}_1$ and $w_{0_J}$ such that $\inner{w_{0_J}, A_{J \bullet} h} = -\norm{A_{J \bullet} h}_1$, we can rewrite
\begin{equation}
\label{eqt::equation of v_0 and w_0}
\begin{split}
\inner{v_0, h} &- \lambda \inner{w_0, A_{\Omega \bullet} h} = \inner{\sgn(x^{\star}_T), h_T}  \\
&{ }- \lambda \inner{\sgn(e^{\star}_S), A_{S \bullet} h} + \norm{h_{T^c}}_1 + \lambda \norm{A_{J \bullet} h}_1.
\end{split}
\end{equation}

\noindent In addition, the identity relation $z^{(x)} = \lambda A^*_{ \Omega \bullet} z^{(e)}$ can be reformulated as
\begin{align*}
\left(
  \begin{array}{c}
    \sgn(x^{\star}_T) \\
    0_{T^c} \\
  \end{array}
\right)
&- \lambda A^*_{\Omega \bullet } \left(
  \begin{array}{c}
    \sgn(e^{\star}_S) \\
    0_J \\
  \end{array}
\right) \\
&=
-\left(
  \begin{array}{c}
    0_T \\
    z^{(x)}_{T^c} \\
  \end{array}
\right)
+
\lambda A^*_{\Omega \bullet } \left(
  \begin{array}{c}
    0_S \\
    z^{(e)}_J \\
  \end{array}
\right).
\end{align*}

\noindent Taking the inner product with $h$ on both sides yields
\begin{align*}
\inner{\sgn(x^{\star}_T), h_T} &- \lambda \inner{A^*_{S \bullet} \sgn(e^{\star}_S), h} \\
&= - \inner{z^{(x)}_{T^c}, h_{T^c}} + \lambda \inner{A^*_{J \bullet } z^{(e)}_J, h} \\
&= - \inner{z^{(x)}_{T^c}, h_{T^c}} + \lambda \inner{z^{(e)}_J, A_{J \bullet} h}.
\end{align*}

\noindent Notice also that $\inner{A^*_{S \bullet } \sgn(e^{\star}_S), h} = \inner{\sgn(e^{\star}_S), A_{S \bullet} h}$. Hence, (\ref{eqt::equation of v_0 and w_0}) is equivalent to
\begin{align*}
&\inner{v_0, h} - \lambda \inner{w_0, A_{\Omega \bullet} h} \\
&= \norm{h_{T^c}}_1 + \lambda \norm{A_{J \bullet} h}_1 - \inner{z^{(x)}_{T^c}, h_{T^c}} + \lambda \inner{z^{(e)}_J, A_{J \bullet} h} \\
%
%
&\geq \frac{1}{4} (\norm{h_{T^c}}_1 + \lambda \norm{A_{J \bullet} h}_1),
\end{align*}
where the last inequality is due to $\inner{z^{(x)}_{T^c}, h_{T^c}} \leq \norm{z^{(x)}_{T^c}}_{\infty} \norm{h_{T^c}}_1 \leq \frac{3}{4} \norm{h_{T^c}}_1$ and $ \inner{z^{(e)}_J, A_{J \bullet} h} \leq \norm{z^{(e)}_J}_{\infty} \norm{A_{J \bullet} h}_1 \leq \frac{3}{4} \norm{A_{J \bullet} h}_1$. Substituting this inequality into (\ref{inq::bound ||x*+h||_1 + lambda ||e*+f||_1 }), we complete the proof.
\end{proof}

From the result of Lemma \ref{lem::construction of dual certificate (z, e) - Sub orthogonal matrix - tradition}, in order to prove exact recovery of the convex program, it suffices to construct a dual certificate $( z^{(x)}, z^{(e)} )$ obeying the conditions of Lemma \ref{lem::construction of dual certificate (z, e) - Sub orthogonal matrix - tradition}. Partitioning $z^{(x)}$, $z^{(e)}$ into two subsets belonging to $T$ and $T^c$, $S$ and $S^c$, the identity relation between $z^{(x)}$ and $z^{(e)}$ can be reformulated as follows
\begin{equation}
\label{eqt::relax relation between z^x and z^e - part 2 - modified}
\begin{split}
&\left(
  \begin{array}{c}
    \sgn(x^{\star}_T) \\
    0_{T^c} \\
  \end{array}
\right)
+
\left(
  \begin{array}{c}
    0_T \\
    z^{(x)}_{T^c} \\
  \end{array}
\right) \\
&= \lambda A^* \left(
  \begin{array}{c}
    \sgn(e^{\star}_S) \\
    0_J \\
    0_{\Omega^c} \\
  \end{array}
\right)
+
\lambda A^* \left(
  \begin{array}{c}
    0_S \\
    z^{(e)}_J \\
    0_{\Omega^c}
  \end{array}
\right).
\end{split}
\end{equation}

If we can construct a pair of vectors $(v^{(x)}, w^{(e)})$ such that $v^{(x)} + w^{(e)}$ is equal to both sides of (\ref{eqt::relax relation between z^x and z^e - part 2 - modified}), that is
$$
\left\{
  \begin{array}{ll}
    v^{(x)}_T + w^{(e)}_T = \sgn(x^{\star}_T) \\
    v^{(x)}_{T^c} + w^{(e)}_{T^c} = z^{(x)}_{T^c}\\
    A_{S \bullet} (v^{(x)} + w^{(e)}) = \lambda \sgn(e^{\star}_S) \\
    A_{J \bullet} (v^{(x)} + w^{(e)}) = \lambda z^{(e)}_J  \\
    A_{\Omega^c \bullet} (v^{(x)} + w^{(e)}) = 0,
  \end{array}
\right.
$$
then the existence of the dual certificate $(z^{(x)}, z^{(e)})$ in Lemma \ref{lem::construction of dual certificate (z, e) - Sub orthogonal matrix - tradition} is guaranteed. As a consequence, it now suffices to produce a dual pair $(v^{(x)}, w^{(e)})$ obeying
\begin{equation}
\label{eqt::requirements of the dual pair v^x and w^e - final}
\left\{
  \begin{array}{ll}
    v^{(x)}_T = \sgn(x^{\star}_T)   \\
    \norm{v^{(x)}_{T^c}}_{\infty} < 3/8  \\
    A_{S \bullet} v^{(x)} = 0  \\
    \norm{A_{J \bullet} v^{(x)}}_{\infty} < 3 \lambda/8    \\
    A_{\Omega^c \bullet} v^{(x)} = 0
  \end{array}
\right.
\quad \text{and} \quad
\left\{
  \begin{array}{ll}
	w^{(e)}_T = 0 \\
	\norm{w^{(e)}_{T^c}}_{\infty} < 3/8 \\
       A_{S \bullet}w^{(e)} = \lambda \sgn(e^{\star}_S)\\
	 \norm{ A_{J \bullet}w^{(e)}}_{\infty} < 3 \lambda /8  \\
	A_{\Omega^c \bullet} w^{(e)} = 0.
  \end{array}
\right.
\end{equation}

In the next section, we will establish that the valid dual pair ($v^{(x)}, w^{(e)}$) exists with probability converging to unity.

\subsection{Dual certification constructions}

We now propose to construct a dual certificate pair ($v^{(x)}, w^{(e)}$) whose components are described as follows

\begin{enumerate}

	\item \textit{Construction of $w^{(e)}$ via least-square.} Since $w^{(e)}_T = 0$, the identity conditions $A_{S \bullet}w^{(e)} = \lambda \sgn(e^{\star}_S)$ and $A_{\Omega^c \bullet} w^{(e)} = 0$ can now be represented by a single equation
\begin{equation}
\label{eqt::linear equations of w^e_T^c}
A_{J^c T^c} w^{(e)}_{T^c} = \lambda \left(
                                  \begin{array}{c}
                                    \sgn(e^{\star}_S) \\
                                    0_{\Omega^c} \\
                                  \end{array}
                                \right),
\end{equation}
where we recall that $J^c = S \cup \Omega^c $. Next, assuming that $\norm{A_{J^c T}} < 1$, then we have $\norm{I - A_{J^c T^c} A^*_{J^c T^c }} = \norm{A_{J^c T} A^*_{ J^c T}} < 1$. Consequently, matrix $A_{J^c T^c} A^*_{ J^c T^c}$ is invertible. We then set
\begin{equation}
\label{eqt::construction of w^e_T^c - part 2}
w^{(e)}_{T^c} = \lambda A^*_{ J^c T^c} (A_{J^c T^c} A^*_{ J^c T^c})^{-1} \left(
                                  \begin{array}{c}
                                    \sgn(e^{\star}_S) \\
                                    0_{\Omega^c} \\
                                  \end{array}
                                \right).
\end{equation}

Clearly, $w^{(e)}_{T^c}$ is the least-square solution of the linear system in (\ref{eqt::linear equations of w^e_T^c}). This construction has a natural interpretation: among all solutions of the linear system, $w^{(e)}_{T^c}$ has the minimum $\ell_2$-norm. We expect that its $\ell_{\infty}$-norm is also sufficiently small to obey the condition in (\ref{eqt::requirements of the dual pair v^x and w^e - final}).
	
	\item \textit{Construction of $v^{(x)}$.}  A simple way to produce $v^{(x)}$ is as follows
\begin{equation}
\label{eqt::construction of w^x_T^c - part 2}
    v^{(x)} = A^*_{J \bullet } A_{J T} (A^*_{J T} A_{J T})^{-1} \sgn(x^{\star}_T).
\end{equation}

It is obvious from this construction that $v^{(x)}_T = \sgn(x^{\star}_T)$. Furthermore, $A_{S \bullet} v^{(x)} = 0$ and $A_{\Omega^c \bullet} v^{(x)} = 0$ due to the orthogonality property of the matrix $A$. Thus, all three identity relations with respect to $v^{(x)}$ in (\ref{eqt::requirements of the dual pair v^x and w^e - final}) are guaranteed.

\end{enumerate}

We now state two key lemmas that establish the $\ell_{\infty}$-norm bounds for $v^{(x)}$ and $w^{(e)}$.

\begin{lem}
\label{lem::prove bounds of v^x}
Assume that $\Omega \sim \Ber(\eta)$ and $S \sim \Ber(\eta \rho)$ where parameters $\eta = m/n$ and $\rho = s/m$. Under the same assumptions as in Theorem \ref{thm::detail of the main theorem}, with high probability, the dual vector $v^{(x)}$ constructed in (\ref{eqt::construction of w^x_T^c - part 2}) obeys
\begin{enumerate}
	\item $\norm{A_{J \bullet} v^{(x)} }_{\infty} <  3 \lambda/8$,
	\item $\norm{v^{(x)}_{T^c}}_{\infty} < 3/8$.
\end{enumerate}
\end{lem}

\begin{lem}
\label{lem::prove bounds of w^e}
Assume that $\Omega$ and $S$ are sampled as in Lemma \ref{lem::prove bounds of v^x}. Under the same assumptions as in Theorem \ref{thm::detail of the main theorem}, with high probability, the dual vector $w^{(e)}$ constructed in  (\ref{eqt::construction of w^e_T^c - part 2}) obeys
\begin{enumerate}
	\item $\norm{w^{(e)}_{T^c}}_{\infty} < 3/8 $,
	\item $\norm{A_{J \bullet} w^{(e)} }_{\infty} < 3 \lambda/8$.
\end{enumerate}
\end{lem}

Lemmas \ref{lem::prove bounds of v^x} and \ref{lem::prove bounds of w^e} suggest the existence of ($z^{(x)}, z^{(e)}$). In other words, the solution of the convex program in (\ref{opt::l1-l1 minimization - sub orthogonal matrix}) is exact and unique.

\section{Proofs of dual certificates}
\label{sec::dual certificate}


\subsection{Important auxiliary lemmas}

In this section, we first develop several auxiliary results concerning the main proof.

\begin{lem}
\label{lem::show the invertibility of A*_ST A_ST}
Let $S_0$ be locations sampled randomly from the set $\{1,...,n\}$, $S_0 \sim \Ber(\rho_0)$. With probability of success at least $1 - n^{-1}$, we have
$$
\norm{ I_{k \times k} - \rho_0^{-1} A^*_{S_0 T} A_{S_0 T} } \leq \epsilon,
$$
provided that $\rho_0 \geq C_0 \frac{\epsilon^{-2} \mu k \log n}{n}$ for $C_0 = 2^{3/4} e^2 \sqrt{\pi e}$.
\end{lem}

This result has been known in the literature \cite{RV_MatSamp_2007_J}, \cite{Tropp_Subdic_2008_J}, \cite{NDT_LRMA_2009_C}. However, for completeness, we provide a brief proof which relies on high order moment bound of the spectral norm. We emphasize that the lemma is important since it provides us the bound of $\norm{A_{J^c T}}$. In fact, recall that $J \sim \Ber(\rho_0)$ with $\rho_0 = \eta(1 - \rho)$, Lemma \ref{lem::show the invertibility of A*_ST A_ST} suggests that
\begin{equation}
\label{inq::bound spectral norm of I - rho A*_JT A_IT}
\norm{I - \rho_0^{-1} A^*_{JT} A_{J T} } \leq \epsilon,
\end{equation}
provided that $\rho_0 \geq C \epsilon^{-2} \frac{\mu k \log n}{n}$. Furthermore, from the fact that $A^*_{JT} A_{J T} = I - A^*_{J^c T} A_{J^c T}$, we obtain
\begin{align*}
\epsilon &\geq \norm{I - \rho_0^{-1} ( I - A^*_{J^c T} A_{J^c T}) } \\
&\geq \rho_0^{-1} \norm{A^*_{J^c T} A_{J^c T}} - (\rho_0^{-1} - 1).
\end{align*}

\noindent This inequality leads to $\norm{A^*_{J^c T} A_{J^c T}} \leq \rho_0 \epsilon + (1 - \rho_0)$. We conclude the argument by the following proposition.
\begin{prop}
\label{lem::bound spectral norm of A_Jc Tc}
Provided that $m-s \geq 4 C_0 \mu k \log n $. With probability at least $1-n^{-1}$, we have
$$
\norm{A_{J^c T}} \leq \sqrt{1 - \rho_0/2}.
$$
\end{prop}

\begin{proof}[Proof of Lemma \ref{lem::show the invertibility of A*_ST A_ST}]

Define $S_0 = \{i: \delta_i = 1 \}$ where $\delta_i$ is an independent sequence of Bernoulli variables with parameter $\rho_0$ and denote $u_i$ ($i \in S_0$) to be row vectors of $A_{S_0 T}$. With these notation, we have
$$
A^*_{S_0 T} A_{S_0 T} = \sum_{i \in S_0} u_i \otimes u_i = \sum_{i = 1}^n \delta_i u_i \otimes u_i.
$$

\noindent Applying Theorem 5 of \cite{NDT_LRMA_2009_C} with $q = \log n$, we obtain
\begin{align*}
&\left( \E \norm{ I - \rho_0^{-1} \sum_{i = 1}^n \delta_i u_i \otimes u_i }^{\log n}  \right)^{1/\log n} \\
&\leq C \sqrt{\rho_0^{-1} \log n} \max_i \norm{u_i}_2  \\
&\leq C \sqrt{\rho_0^{-1} (\mu k \log n) / n} := E,
\end{align*}
where the constant $C = 2^{3/4} \sqrt{\pi e}$, and the last inequality holds from $\norm{u_i}_2 \leq \sqrt{\mu k / n}$.

\noindent From Markov's inequality, we can establish
\begin{align*}
\Prob \left( \norm{ I - \rho_0^{-1} \sum_{i = 1}^n \delta_i u_i \otimes u_i } \geq \epsilon \right) \leq \frac{E^{\log n}}{(\epsilon)^{\log n}}.
\end{align*}

\noindent By the assumption of the Lemma that $\frac{C \sqrt{\rho_0^{-1} \mu k \log n}}{\epsilon} \leq \frac{1}{e}$, we have with probability of success at least $1 - n^{-1}$,
$$
\norm{ I - \rho_0^{-1} \sum_{i = 1}^n \delta_i u_i \otimes u_i }  \leq \epsilon,
$$
as claimed.
\end{proof}

The next lemma shows the matrix $A_{S_0 T}$ is almost orthogonal to the matrix $A_{S_0 T^c}$ where $S_0$ is a random subset selected from columns of the matrix $A$. This property is important in distinguishing the set $T$ from the set $T^c$ and helping the algorithm identify the true support of $x^{\star}$. We defer the proof to the Appendix.

\begin{lem}
\label{lem::show the almost orthogonality of A_ST^c versus A_ST}
Let $S_0$ be locations sampled randomly from the set $\{1,...,n\}$, $S_0 \sim \Ber(\rho_0)$. With probability at least $1 - 3 n^{-1}$, the following inequality obeys
$$
\norm{ A^*_{S_0 T} u }_{2} \leq  \sqrt{C' \rho_0 \frac{\mu \max \{k ,\log n \}}{n}}
$$
for any column vector $u$ of the matrix $A_{S_0 T^c}$, provided that $\rho_0 \geq C \frac{\mu \max \{ k , \log n\}}{n}$ where $C$ and $C'$ are numerical constants.
\end{lem}

We are now ready to prove Lemmas \ref{lem::prove bounds of v^x} and \ref{lem::prove bounds of w^e} regarding the dual certificates.

\subsection{Proof of Lemma \ref{lem::prove bounds of v^x}}

\begin{proof}[\textbf{  Part 1}] By the construction of $v^{(x)}$ in (\ref{eqt::construction of w^x_T^c - part 2}),
$$
A_{J \bullet} v^{(x)} = A_{JT} (A^*_{JT} A_{JT})^{-1} \sgn(x^{\star}_T).
$$

\noindent Denote $u_i$ as a row of the matrix $A_{JT}$, we have
\begin{align*}
\norm{A_{J \bullet} v^{(x)}}_{\infty} &= \max_i | u_i  (A^*_{JT} A_{JT})^{-1} \sgn(x^{\star}_T) | \\
&:= \max_i | \inner{W u^*_i, \sgn(x^{\star}_T)} |,
\end{align*}
where we denote $W := (A^*_{JT} A_{JT})^{-1}$. The right-hand side is a sum of zero mean random variables, which can be bounded by Hoeffding's inequality. Hence,
$$
\Prob \left( | \inner{W u^*_i, \sgn(x^{\star}_T)} | \geq \tau \right) \leq 2 \exp \left( - \frac{\tau^2}{2 \norm{W u^*_i}_2^2 }  \right).
$$

\noindent Notice from (\ref{inq::bound spectral norm of I - rho A*_JT A_IT}) that with probability converging to one, $\norm{I - \rho_0 A^*_{JT} A_{J T}} \leq \epsilon$. Thus, $(1 - \epsilon) \rho_0 \leq \sigma_{\min} (A^*_{JT} A_{J T}) \leq \sigma_{\max} (A^*_{JT} A_{J T}) \leq (1 + \epsilon) \rho_0$ where $\sigma_{\min}$ and $\sigma_{\max}$ are minimum and maximum singular values of the matrix. In addition, we have exploited the fact that spectral norm for any matrix $H$ obeys $\norm{H^{-1}} \leq \frac{1}{\sigma_{\min} (H)}$. Thus, conditioning on the event $\oper E = \{ \norm{I - \rho_0 A^*_{JT} A_{J T}} \leq \epsilon \}$, we have
$$
\norm{W} \leq \frac{1}{(1 - \epsilon) \rho_0} \leq \frac{2}{\rho_0},
$$
with the choice of $\epsilon \leq 1/2$. Consequently, combining with $\norm{u_i}_2^2 \leq \mu \frac{k}{n}$, we conclude that $\norm{W u^*_i}_2^2 \leq \frac{4 \mu k}{\rho_0^2 n}$.

\noindent Now setting $\tau := \sqrt{\frac{16 \mu k \log n}{\rho_0^2 n}}$ and taking the union bound over all row vectors of matrix $A_{J T}$, we obtain
\begin{align*}
\Prob \left( \norm{A_{J \bullet} v^{(x)}}_{\infty} \geq \sqrt{\frac{16 \mu k \log n}{\rho_0^2 n}} \right) &\leq 2 |J| e^{- 2 \log n} + \Prob (\oper E^c) \\
&\leq 3 n^{-1},
\end{align*}
where the inequality follows from the total probability rule: $\Prob (F \geq \tau) \leq \Prob(F \geq \tau | \oper E) + \Prob(\oper E^c)$ with $F := \norm{A_{J \bullet} v^{(x)}}_{\infty}$. We conclude that $\norm{A_{J \bullet} v^{(x)}}_{\infty} \leq \frac{\lambda}{4}$ as long as $k \leq C \frac{\lambda^2 \rho_0^2 n}{\mu \log n}$. Replace $\lambda = \sqrt{\frac{1}{\gamma \log n} \frac{n}{m}}$, $\rho_0 = \frac{m-s}{n}$ and $s = \gamma m$, one can see that the upper bound of $k$ automatically follows from the assumption that $k \leq C \frac{(1-\gamma)^2}{\gamma} \frac{m}{\mu^2 (\log n)^2}$.
\end{proof}

\begin{proof}[\textbf{Part 2}] In this part, we need to show that with high probability,
$$
\norm{A^*_{J T^c} A_{J T} (A_{J T}^* A_{JT})^{-1} \sgn(x^{\star}_T) }_{\infty} \leq 3/8.
$$

\noindent Denote $u_i$ as a column vector of the matrix $A_{J T^c}$ and consider $u^*_i A_{J T} (A_{J T}^* A_{JT})^{-1} \sgn(x^{\star}_T) = \inner{(A_{J T}^* A_{JT})^{-1} A_{J T}^* u_i, \sgn(x_T^{\star}) }$, which is a sum of random variables. Its absolute value can be estimated via Hoeffding's inequality,
$$
\Prob ( |u^*_i A_{J T} (A_{J T}^* A_{JT})^{-1} \sgn(x^{\star}_T)| \geq \tau ) \leq 2 \exp \left( - \frac{\tau^2}{2 \norm{z}^2_2 } \right),
$$
where $z := (A_{J T}^* A_{JT})^{-1} A_{J T}^* u_i$. As previously showed, conditioning on the event $\oper E_1 = \{ \norm{I - \rho_0 A^*_{JT} A_{J T}} \leq \epsilon \leq 1/2 \}$, we have $\norm{(A_{J T}^* A_{JT})^{-1}} \leq 2/\rho_0$. In addition, we define the event $\oper E_2 := \{ \norm{ A^*_{J T} u_i }_{2} \leq  \sqrt{C' \rho_0 \frac{\mu \max \{k ,\log n \}}{n}} \}$, which bounds the $\ell_2$ norm of $A_{J T}^* u$ with $J \sim \Ber (\rho_0)$. We showed from Lemma \ref{lem::show the almost orthogonality of A_ST^c versus A_ST} that $\Prob (\oper E_2) \leq 1-3n^{-1}$. Therefore, conditioning on both $\oper E_1$ and $\oper E_2$, we get
$$
\norm{z}_2 \leq \norm{ (A_{J T}^* A_{JT})^{-1}} \norm{A_{J T}^* u_i} \leq \sqrt{C'  \frac{\mu \max \{k ,\log n \}}{\rho_0 n}}.
$$

\noindent Setting $\tau^2 := 4 C' \frac{\mu \max \{k ,\log n \}}{\rho_0 n}$ and taking the union bound, we conclude that
$$
\Prob \left( \norm{v^{(x)}_{T^c}}_{\infty} \geq \tau \right) \leq 2 (n-k) e^{- 2 \log n} + \Prob (\oper E_1^c) + \Prob (\oper E_2^c),
$$
which is less than $6 n^{-1}$. Now replace $\rho_0 = \frac{m-s}{n}$ and assume that $m - s \geq C \mu \max \{ k , \log n\} \log n$ where $C = 4 (8/3)^2 C'$, we achieve $\norm{v^{(x)}_{T^c}}_{\infty} \leq 3/8$ as claimed.

\end{proof}

\subsection{Proof of Lemma \ref{lem::prove bounds of w^e}}

\subsubsection{Preliminary results}

In order to set up the bounds of Lemma \ref{lem::prove bounds of w^e}, it is necessary to estimate the spectral norm bound of $\norm{A_{S T}}$. The following proposition establishes such a bound.

\begin{prop}
\label{prop::spectral norm bound of A_TS}
With probability at least $1 - n^{-1}$,
$$
\norm{A^*_{ST} A_{S T}} \leq \left(1 + \sqrt{\frac{\mu k \log n}{s}} \right) \eta \rho.
$$
\end{prop}

\begin{proof}

Recall that $S \sim \Ber (\eta \rho)$. By Lemma \ref{lem::show the invertibility of A*_ST A_ST}, with high probability, we have
\begin{equation}
\label{eqt::bound norm of I - rho A*_ST A_ST}
\norm{I - (\eta \rho)^{-1} A^*_{ST} A_{S T}} \leq \epsilon_1,
\end{equation}
provided $\eta \rho \geq C_0 \frac{\epsilon^{-2}_1 \mu k \log n}{n}$. Note that $\eta \rho =  \frac{s}{n}$, thus the condition is equivalent to $s \geq \epsilon^{-2}_1 \mu k \log n$. This inequality is automatically satisfied by setting $\epsilon_1 = \sqrt{\frac{\mu k \log n}{s}}$. Therefore, (\ref{eqt::bound norm of I - rho A*_ST A_ST}) gives us
\begin{align*}
\norm{A^*_{ST} A_{S T}} \leq \left( 1 + \sqrt{\frac{\mu k \log n}{s}} \right) \eta \rho
\end{align*}
as claimed.
\end{proof}

\subsubsection{Main proofs}

\begin{proof} [\textbf{ Part 1}]  We will start with the construction of $w^{(e)}_{T^c}$ in (\ref{eqt::construction of w^e_T^c - part 2}). Our goal is to show that with high probability,
$$
V := \lambda \norm{A^*_{J^c T^c } (A_{J^c T^c} A^*_{J^c T^c })^{-1} \left(
                                  \begin{array}{c}
                                    \sgn(e^{\star}_S) \\
                                    0_{\Omega^c} \\
                                  \end{array}
                                \right)}_{\infty} < 1/4.
$$

\noindent By series expansion, $(I-H)^{-1} = I + \sum_{j=1}^{\infty} H^j$. We have
\begin{align*}
A^*_{J^c T^c } (A_{J^c T^c} A^*_{ J^c T^c})^{-1} &= A^*_{J^c T^c } (I - A_{J^c T} A^*_{J^c T })^{-1}\\
& = A^*_{J^c T^c } + \sum_{j\geq 1} A^*_{J^c T^c } (A_{J^c T} A^*_{J^c T})^j.
\end{align*}

\noindent Toward that end, denote $r := \left(
                                  \begin{array}{c}
                                    \sgn(e^{\star}_S) \\
                                    0_{\Omega^c} \\
                                  \end{array}
                                \right)$. To establish the upper bound of $V$, we elaborate on the $\ell_{\infty}$-norms of two quantities relating to summands of the series expansion. The bound of $V$ is then followed from the triangular inequality. For the first term $V_1 := \lambda \norm{A^*_{ J^c T^c} r}_{\infty}$, we have
$$
V_1 = \lambda \norm{A^*_{S T^c} \sgn(e^{\star}_S) }_{\infty} = \lambda \max_i |\inner{u_i, \sgn(e_S)}|,
$$
where $u_i$ is denoted as a column vector of $A_{S T^c}$. We notice that $\inner{u_i, \sgn(e^{\star}_S)}$ is a sum of zero mean random variables (by the random assumption on the sign of $e_S$). Applying Hoeffding's inequality yields
\begin{align*}
\Prob \left( | \inner{u_i, \sgn(e^{\star}_S)} | \geq \tau\right) &\leq 2 \exp \left( - \frac{2 \tau^2}{4 \norm{u_i}_2^2} \right) \\
&\leq 2 \exp \left( -\frac{\tau^2}{2} \frac{n}{\mu s}  \right),
\end{align*}
where the last inequality is due to $\norm{u_i}_2^2 \leq \frac{\mu s}{n}$. Next, choosing $\tau = \frac{1}{8 \lambda}$ and taking the union bound over all $i \in T^c$ yield
$$
\Prob \left(  \lambda \norm{A^*_{J^c T^c } r}_{\infty} \geq \frac{1}{8} \right) \leq \exp \left( - \frac{n}{128 \mu \lambda^2 s} + \log (2n)\right),
$$
which is bounded away by $e^{-\log n} = n^{-1}$ as long as $s \leq C \frac{n}{\mu \lambda^2 \log n} = C \gamma m$.

For the remainder term, denote the quantity $V_r := \lambda \norm{\sum_{j \geq 1} A^*_{J^c T^c } (A_{J^c T} A^*_{J^c T})^j r}_{\infty}$, we have
\begin{align*}
V_r &= \lambda \norm{\sum_{j \geq 1} (A^*_{J^c T^c } A_{J^c T}) (A^*_{J^c T} A_{J^c T})^{j-1} A^*_{J^c T} r}_{\infty} \\
&= \lambda \norm{\sum_{j \geq 1} (A^*_{J T^c } A_{J T}) (A^*_{ J^c T} A_{J^c T})^{j-1} A^*_{ST} \sgn(e^{\star}_S)}_{\infty} \\
&= \lambda \max_{i \in T^c} \left| \sum_{j \geq 0} u^*_i A_{J T} (A^*_{J^cT} A_{J^c T})^j A^*_{ ST} \sgn(e^{\star}_S) \right|,
\end{align*}
where $u_i$ is denoted as the $i^{th}$ column vector of $A_{J T^c}$. Notice that vector $u_i$ has length $(m-s)$.

Let $W :=  \sum_{j \geq 0} A_{J T} (A^*_{ J^c T} A_{J^c T})^j A^*_{ S T}$. We consider the term inside the max function $V_i = | \inner{W^* u_i, \sgn(e^{\star}_S)} |$. Again, this quantity's bound is an application of Hoeffding's inequality,
$$
\Prob \left( | \inner{W^* u_i, \sgn(e^{\star}_S)} | \geq \tau\right) \leq 2 \exp \left( - \frac{2 \tau^2}{4 \norm{W^* u_i}_2^2} \right)
$$

\noindent Next, we have
\begin{align*}
\norm{W^* u_i } &\leq \norm{A_{S T}} \left( \sum_{j \geq 0} \norm{A^*_{T J^c} A_{J^c T}}^j \right) \norm{A^*_{T J}}  \norm{u_i}_2 \\
&= \frac{\norm{A_{S T}} \norm{A^*_{ J T}}}{1 - \norm{A^*_{J^c T} A_{J^c T}}} \norm{u_i}_2.
\end{align*}

\noindent We now provide the spectral and $\ell_2$ norms of these terms. Define the following three events
$$
\oper E_1 := \{ \norm{A^*_{J^c T} A_{J^c T}} \leq 1 - \rho_0/2 \},
$$
$$
\oper E_2 := \{ \norm{A^*_{J T }} \leq \sqrt{3 \rho_0/2} \}, \quad \text{and}
$$
$$
\oper E_3 := \{ \norm{A_{S T}} \leq (1 + \sqrt{\frac{\mu k \log n}{s}})^{1/2} \sqrt{ \eta \rho} \}.
$$
Recall by Proposition \ref{lem::bound spectral norm of A_Jc Tc} that the event $\oper E_1$ occurs with high probability. Moreover, from Lemma \ref{lem::show the invertibility of A*_ST A_ST}, with high probability $\norm{ I_{k \times k} - \rho^{-1}_0 A^*_{JT} A_{J T}  } \leq \epsilon$ provided $\rho_0 \geq C \epsilon^{-2} \frac{\mu k \log n}{n}$. Thus, $\norm{A^*_{JT}} \leq \sqrt{\rho_0 (1+\epsilon)} \leq \sqrt{3 \rho_0/2} $, assuming that $\epsilon \leq 1/2$. Finally, $\oper E_3$ occurs by Proposition \ref{prop::spectral norm bound of A_TS} and the fact that $\norm{u_i}_2^2 \leq \frac{\mu (m-s)}{n} = \mu \rho_0$. Conditioning on these events, we conclude that
\begin{align*}
\norm{W^* u_i}_2^2 &\leq \frac{(3 \rho_0/2) (\mu \rho_0) (\eta \rho)}{(\rho_0/2)^2}  \left(1 + \sqrt{\frac{\mu k \log n}{s}} \right)  \\
&\leq 6 \mu \frac{s}{n} \left(1 + \sqrt{\frac{\mu k \log n}{s}} \right).
\end{align*}

We consider two following cases regarding the size of the set $S$:

\textit{Case 1:} if $s \geq \mu k \log n$, then $\norm{W^* u_i}_2^2 \leq 12 \mu \frac{s}{n}$. Set $\tau = \frac{1}{8 \lambda}$ and take the union bound over all $i \in T^c$, we attain
\begin{align*}
\Prob \left( V_r \geq \frac{1}{8 \lambda}   \right) &\leq 2 \exp \left( -\frac{1}{256 \lambda^2 \mu \frac{s}{n}}  + \log n \right) \\
&{ }+ \Prob (\oper E_1^c) + \Prob (\oper E_2^c) + \Prob (\oper E_3^c) .
\end{align*}

\noindent By assuming $s \leq C \gamma m$ with a sufficiently small constant $C$, $\mu \lambda^2 \frac{s}{n} \leq C \frac{1}{\log n}$. Hence, $V_r \leq \frac{1}{8 \lambda}$ with probability $1 - 5 n^{-1}$.

\textit{Case 2:} if $s \leq \mu k \log n$, then $\norm{W^* u_i}_2^2 \leq 12 \mu \frac{s}{n} \sqrt{\frac{\mu k \log n}{s}}$. Again, setting $\tau = \frac{1}{8 \lambda}$ and taking the union bound, we have
\begin{align*}
\Prob \left( V_r \geq \frac{1}{8 \lambda} \right) &\leq 2 \exp \left( -\frac{1}{256 \lambda^2 \mu \frac{s}{n} \sqrt{\frac{\mu k \log n}{s}}}  + \log n \right) \\
&\leq 2 e^{-\log n} = 2 n^{-1},
\end{align*}
provided that $s \leq C \gamma m$ and $k \leq C' \frac{\gamma m}{\mu \log n}$.

We complete the proof by employing the triangular inequality: $V \leq V_1 + V_r \leq \frac{1}{4 \lambda}$.
\end{proof}

\begin{proof} [\textbf{ Part 2}]
In this part, we need to show that with high probability
$$
V := \lambda \norm{A_{J T^c } A^*_{J^c T^c } (A_{J^c T^c} A^*_{ J^c T^c})^{-1} \left(
                                  \begin{array}{c}
                                    \sgn(e^{\star}_S) \\
                                    0_{\Omega^c} \\
                                  \end{array}
                                \right)}_{\infty} < \frac{\lambda}{4}.
$$

Again by series expansion, we first have $(A_{J^c T^c} A^*_{ J^c T^c})^{-1} = \sum_{j \geq 0} (A_{J^c T} A^*_{J^c T})^j$. Moreover, since $A_{J T^c } A^*_{ J^c T^c} = -A_{J T} A^*_{J^c T}$, we arrive at
\begin{align*}
&A_{J T^c } A^*_{ J^c T^c} (A_{J^c T^c} A^*_{ J^c T^c})^{-1} \left(
                                  \begin{array}{c}
                                    \sgn(e^{\star}_S) \\
                                    0_{\Omega^c} \\
                                  \end{array}
                                \right) \\
&= \sum_{j \geq 0} A_{J T} A^*_{J^c T} (A_{J^c T} A^*_{J^c T})^j \left(
                                  \begin{array}{c}
                                    \sgn(e^{\star}_S) \\
                                    0_{\Omega^c} \\
                                  \end{array}
                                \right) \\
&= \sum_{j \geq 0} A_{J T} (A^*_{J^c T} A_{J^c T})^j A^*_{ST} \sgn(e^{\star}_S).
\end{align*}

\noindent Let $W := \sum_{j\geq 0} (A^*_{J^c T} A_{J^c T})^j A^*_{ S T}$ and let $u_i \in \R^k$ be a row vector of $A_{J T}$. We consider the following bound $V_i := \left| \inner{W^* u^*_i, \sgn(e^{\star}_S)} \right|$. Analogous to the preceding proofs, Hoeffding's inequality is used to estimate $V_i$,
$$
\Prob \left( V_i \geq \tau  \right) \leq 2 \exp \left( - \frac{2 \tau^2}{4 \norm{W^* u^*_i}_2^2} \right).
$$

\noindent The spectral norm of $W$ can now be estimated as follows
$$
\norm{W} \leq \norm{A^*_{ST}} (\sum_{j\geq 0} \norm{A^*_{ J^c T} A_{J^c T}}^j) = \frac{\norm{A^*_{ST}}}{1 - \norm{A^*_{J^c T} A_{J^c T}}}.
$$

\noindent Conditioning on events $\oper E_1$ and $\oper E_3$ in Part 1, together with $\norm{u_i}_2 \leq \sqrt{\mu \frac{k}{n}}$, we get
$$
\norm{W^* u^*_i}_2 \leq \norm{W} \norm{u_i}_2 \leq \sqrt{ 4 \frac{2 \eta \rho }{\rho_0^2} \frac{\mu k}{n}}.
$$

\noindent Set $\tau = 1/4$ and take the union bound over all $i \in J$,
$$
\Prob (V \geq 1/4 \text{ }|\text{ } \oper E_1, \oper E_3) \leq 2 \exp \left( - \frac{\rho_0^2 n}{256 \mu \eta \rho k} + \log n \right).
$$

\noindent The right-hand side is less than $2 e^{-\log n} = 2 n^{-1}$ as long as $\frac{\rho_0^2 n}{256 \mu \eta \rho k} = \frac{(m-s)^2}{256 \mu s k} \geq 6 \log n$. This is automatic from the assumptions that $k \leq C \frac{(1-\gamma)^2}{\gamma} \frac{m}{\mu \log n}$ and $s \leq \gamma m$.
\end{proof}

\section{Proof of Theorems \ref{thm::main theorem - with dense noise and without sparse noise} and \ref{thm::main theorem - with dense noise and sparse noise}: Dealing with both sparse and dense errors}
\label{sec::proof of last 2 theorems}

\subsection{Proof of Theorem \ref{thm::main theorem - with dense noise and without sparse noise}}

Our proof technique is adapted from \cite{CP_MatCompletion_2009_J} (see also \cite{ZWLCM_StablePCP_2010_C}) but in a different context. In \cite{CP_MatCompletion_2009_J}, the authors studied the matrix completion problem under noisy observations, while we consider the conventional compressed sensing case. Let $\widehat{x}$ be the optimal solution of (\ref{opt::l1-l1 minimization with noise and without sparse noise - sub orthogonal matrix}). Since $x^{\star}$ is also a feasible solution of (\ref{opt::l1-l1 minimization with noise and without sparse noise - sub orthogonal matrix}), $\norm{A_{\Omega \bullet} x^{\star} - b}_2 \leq \sigma$. We have an important observation
\begin{equation}
\label{inq::first observation}
\norm{A_{\Omega \bullet} (\widehat{x} - x^{\star})}_2 \leq \norm{A_{\Omega \bullet} \widehat{x} - b}_2 + \norm{A_{\Omega \bullet} x^{\star} - b}_2 \leq 2 \sigma.
\end{equation}

\noindent Denote $g = \widehat{x} - x^{\star}$, our goal is to establish a bound for $\norm{g}_2$. At first, note that $\norm{g}^2_2 = \norm{A g}^2_2$, the triangular inequality gives us
\begin{equation}
\begin{split}
\label{inq::bound l2-norm of g - first step}
\norm{g}^2_2 = \norm{A_{\Omega \bullet} g}^2_2 + \norm{ A_{\Omega^c \bullet} g}^2_2 = 4 \sigma^2 + \norm{ A_{\Omega^c \bullet} g}_2^2.
\end{split}
\end{equation}

It now remains to bound the second term. Our strategy is to bound $\norm{A^*_{\Omega^c T } A_{\Omega^c \bullet} g}_2$ and $\norm{A^*_{\Omega^c T^c } A_{\Omega^c \bullet} g}_2$ separately, then the bound of $\norm{ A_{\Omega^c \bullet} g}^2_2$ is obtained via the following expression
\begin{equation}
\label{inq::expression of l2 A_omegaC g}
\begin{split}
\norm{ A_{\Omega^c \bullet} g}^2_2 &= \norm{A^*_{\Omega^c \bullet} A_{\Omega^c \bullet} g}^2_2 \\
&= \norm{A^*_{\Omega^c T } A_{\Omega^c \bullet} g}_2^2 + \norm{A^*_{\Omega^c T^c } A_{\Omega^c \bullet} g}_2^2,
\end{split}
\end{equation}
where the first expression follows from $ \norm{A^*_{ \Omega^c \bullet} A_{\Omega^c \bullet} g}_2^2 = \inner{A_{\Omega^c \bullet} g, A_{\Omega^c \bullet} A^*_{ \Omega^c \bullet} A_{\Omega^c \bullet} g} = \inner{A_{\Omega^c \bullet} g, A_{\Omega^c \bullet} g} = \norm{ A_{\Omega^c \bullet} g}^2_2$.

To bound $\norm{A^*_{\Omega^c T^c } A_{\Omega^c \bullet} g}_2$, we bring Lemma \ref{lem::construction of dual certificate (z, e) - Sub orthogonal matrix - tradition} into action: for any perturbation pair ($h,0$) satisfying $A_{\Omega \bullet} h = 0$, we have
\begin{equation}
\norm{x^{\star} + h}_1 \geq \norm{x^{\star}}_1  + \frac{1}{4} \norm{h_{T^c}}_1.
\end{equation}

\noindent By setting $h = A^*_{\Omega^c \bullet } A_{\Omega^c \bullet} g$, we see that $A_{\Omega \bullet} h = 0$. Hence, applying Lemma \ref{lem::construction of dual certificate (z, e) - Sub orthogonal matrix - tradition} yields
\begin{equation}
\label{inq::bound l1 norm of x + A*_Omega^c A_Omega^c g}
\norm{x^{\star} + A^*_{\Omega^c \bullet } A_{\Omega^c \bullet} g}_1 \geq \norm{x^{\star}}_1 + \frac{1}{4} \norm{A^*_{\Omega^c T^c } A_{\Omega^c \bullet} g}_1.
\end{equation}

\noindent Furthermore, noting that $x^{\star} + g$ is the optimal solution of the convex program (\ref{opt::l1-l1 minimization with noise and without sparse noise - sub orthogonal matrix}). This yields
$$
\norm{x^{\star}}_1 \geq \norm{x^{\star} + g}_1 \geq \norm{x^{\star} + A^*_{ \Omega^c \bullet} A_{\Omega^c \bullet} g}_1 - \norm{A^*_{\Omega \bullet} A_{\Omega \bullet} g}_1.
$$

\noindent  In combination with (\ref{inq::bound l1 norm of x + A*_Omega^c A_Omega^c g}), we have an important inequality: $\norm{A^*_{ \Omega^c T^c} A_{\Omega^c \bullet} g}_1 \leq 4  \norm{A^*_{ \Omega \bullet} A_{\Omega \bullet} g}_1$. Since the $\ell_1$-norm dominates the $\ell_2$-norm, $\norm{A^*_{ \Omega^c T^c} A_{\Omega^c \bullet} g}_2 \leq \norm{A^*_{\Omega^c T^c } A_{\Omega^c \bullet} g}_1$ and we have
\begin{equation}
\begin{split}
\label{inq::bound l2-norm of A*_T^c Omega^c u}
\norm{A^*_{ \Omega^c T^c} A_{\Omega^c \bullet} g}_2 &\leq 4 \norm{A^*_{ \Omega \bullet} A_{\Omega \bullet} g}_1 \\
& \leq 4 \sqrt{n} \norm{A^*_{ \Omega \bullet } A_{\Omega \bullet} g}_2 \\
&= 4 \sqrt{n} \norm{A_{\Omega \bullet} g}_2.
\end{split}
\end{equation}

It is left to develop a bound for $\norm{A^*_{\Omega^c T} A_{\Omega^c \bullet} g}_2$. We observe that $A_{\Omega T} A^*_{\Omega^c T} = - A_{\Omega T^c} A^*_{\Omega^c T^c}$ due to the orthogonality property of $A$. Thus, for any vector $u$,
\begin{align*}
\norm{A_{\Omega T} A^*_{\Omega^c T} u}_2 &= \norm{A_{\Omega T^c} A^*_{\Omega^c T^c } u}_2  \\
&\leq \norm{A_{\bullet T^c} A^*_{\Omega^c T^c } u}_2 = \norm{ A^*_{\Omega^c T^c } u}.
\end{align*}

\noindent In addition, applying Lemma \ref{lem::show the invertibility of A*_ST A_ST} with $\Omega \sim \Ber (\eta)$ we have $\norm{I - \eta^{-1} A^*_{\Omega T} A_{\Omega T}} \leq \epsilon \leq 1/2$ with high probability. Hence,
\begin{align*}
&\eta^{-1} \norm{A_{\Omega T} A^*_{\Omega^c T} u}_2^2 \\
&= \eta^{-1} \inner{A_{\Omega T} A^*_{\Omega^c T} u, A_{\Omega T} A^*_{\Omega^c T} u}  \\
&= \eta^{-1} \inner{ A^*_{\Omega^c T} u, A^*_{\Omega T} A_{\Omega T} A^*_{\Omega^c T} u } \\
&= \inner{A^*_{\Omega^c T} u, A^*_{ \Omega^c T} u} - \inner{A^*_{ \Omega^c T} u, (I - \eta^{-1} A^*_{ \Omega T} A_{\Omega T}) A^*_{ \Omega^c T} u} \\
&\geq \norm{A^*_{ \Omega^c T} u}_2^2 - \norm{I - \eta^{-1} A^*_{ \Omega T} A_{\Omega T}} \norm{A^*_{ \Omega^c T} u}_2^2 \\
&\geq \frac{1}{2}\norm{A^*_{ \Omega^c T} u}_2^2.
\end{align*}

\noindent In other words, $ \sqrt{\eta/2}\norm{A^*_{\Omega^c T} u}_2 \leq \norm{A_{\Omega T} A^*_{ \Omega^c T} u}_2$. Combining these pieces together while setting $u = A_{\Omega^c \bullet} g$ yields
\begin{equation*}
\label{inq::bound l2-norm of A*_T Omega^c u}
\norm{A^*_{ \Omega^c T} A_{\Omega^c \bullet} g}_2 \leq \frac{1}{\sqrt{\eta/2}} \norm{A^*_{\Omega^c T^c} A_{\Omega^c \bullet} g}_2.
\end{equation*}

\noindent The right-hand side is in turn bounded by $\frac{4\sqrt{n}}{\sqrt{\eta/2}} \norm{A^*_{\Omega^c T^c}}_2$ due to (\ref{inq::bound l2-norm of A*_T^c Omega^c u}). Inserting this bound with the bound in (\ref{inq::bound l2-norm of A*_T^c Omega^c u}) into (\ref{inq::expression of l2 A_omegaC g}), we obtain
\begin{equation*}
\norm{ A_{\Omega^c \bullet} g}_2^2 \leq \left( \frac{2}{\eta} + 1 \right) 16 n \norm{ A_{\Omega \bullet} g}_2^2 \leq \left( \frac{2}{\eta} + 1 \right) 64 \sigma^2 n ,
\end{equation*}
where the last inequality follows from the known bound in (\ref{inq::first observation}). Combining this result with (\ref{inq::bound l2-norm of g - first step}) we can conclude that
$$
\norm{g}_2 \leq 2 \sigma + 8 \sigma \sqrt{n \left( \frac{2}{\eta} + 1 \right)},
$$
as claimed.

\subsection{Proof of Theorem \ref{thm::main theorem - with dense noise and sparse noise}}

\begin{proof}
The proof of this theorem is considerably more involved since we have to control two residual components $\widehat{x} - x^{\star}$ and $\widehat{e} - e^{\star}$, where $(\widehat{x}, \widehat{e})$ is the optimal solution pair of (\ref{opt::l1-l1 minimization with noise - sub orthogonal matrix}). Set $g^{(x)} = \widehat{x} - x^{\star}$ and $g^{(e)} = \widehat{e} - e^{\star}$, our goal is to bound $\norm{g^{(x)}}_2 + \norm{g^{(e)}}_2$.

At first, notice that $(x^{\star}, e^{\star})$ and $(\widehat{x}, \widehat{e})$ are pairs of feasible solution, we establish an important bound
\begin{equation}
\label{inq::first bound}
\begin{split}
\norm{A_{\Omega \bullet} g^{(x)} + g^{(e)} }_2 &\leq \norm{A_{\Omega \bullet} \widehat{x} + \widehat{e} - b}_2 \\
&{ }+  \norm{A_{\Omega \bullet} x^{\star} + e^{\star} - b}_2 \leq 2 \sigma.
\end{split}
\end{equation}

To bound $ \norm{g^{(x)}}_2 + \norm{g^{(e)}}_2$, we first express $\norm{g^{(x)}}_2$ as $\norm{g^{(x)}}_2^2 = \norm{A g^{(x)}}_2^2 = \norm{A_{\Omega^c \bullet} g^{(x)}}_2^2 + \norm{A_{\Omega \bullet} g^{(x)}}_2^2$. Furthermore, from the fact that $\frac{1}{2}\norm{a + b}^2_2 + \frac{1}{2}\norm{a-b}^2_2 = \norm{a}^2_2 + \norm{b}^2_2$ for any vectors $a$ and $b$, we get
\begin{equation}
\begin{split}
\label{inq::bound l-2 of g^x + l-2 of g^2 - first step}
&\norm{g^{(x)}}^2_2 + \norm{g^{(e)}}^2_2 \\
&= \norm{A_{\Omega^c \bullet} g^{(x)}}^2_2 + \norm{A_{\Omega \bullet} g^{(x)}}^2_2 + \norm{g^{(e)}}^2_2 \\
&\leq \norm{A_{\Omega^c \bullet} g^{(x)}}^2_2 +  \frac{1}{2} \norm{A_{\Omega \bullet} g^{(x)} + g^{(e)}}^2_2 + \frac{1}{2} \norm{A_{\Omega \bullet} g^{(x)} - g^{(e)}}^2_2 \\
&\leq 2 \sigma^2 + \norm{A_{\Omega^c \bullet} g^{(x)}}^2_2 + \frac{1}{2}\norm{A_{\Omega \bullet} g^{(x)} - g^{(e)}}^2_2.
\end{split}
\end{equation}

It is left to bound the sum of the second and third term on the right-hand side of the equation. We express this sum as
\begin{equation*}
\begin{split}
&\norm{A_{\Omega^c \bullet} g^{(x)}}^2_2 +  \frac{1}{2} \norm{A_{\Omega \bullet} g^{(x)} - g^{(e)}}^2_2 \\
&= \norm{A_{\Omega^c \bullet} g^{(x)}}^2_2 + \frac{1}{2} \norm{A_{S \bullet} g^{(x)} - g^{(e)}_S}^2_2 + \frac{1}{2} \norm{A_{J \bullet} g^{(x)} - g^{(e)}_J}^2_2,
\end{split}
\end{equation*}
where we recall that indices in $S$ are locations where measurements are available but unreliable and indices in $J$ are locations where measurements are available and trustworthy and $\Omega = S \cup J$. To upper bound this sum, we consider the establishment of the upper bounds for each term $M_1 := \norm{A_{\Omega^c \bullet} g^{(x)}}^2_2 + \norm{A_{S \bullet} g^{(x)} - g^{(e)}_S}^2_2$ and $M_2 := \norm{A_{J \bullet} g^{(x)} - g^{(e)}_J }^2_2$ separately.

One of the crucial steps in bounding $M_1$ and $M_2$ is the use of Lemma \ref{lem::construction of dual certificate (z, e) - Sub orthogonal matrix - tradition}, which states that for any perturbation pair ($h,f$) satisfying $f=-A_{\Omega \bullet} h$,
\begin{equation*}
\begin{split}
\norm{x^{\star} + h}_1 + \lambda \norm{e^{\star}+f}_1 &\geq \norm{x^{\star}}_1 + \lambda \norm{e^{\star}}_1 \\
&{}+ \frac{\lambda}{4} \norm{f_J}_1 + \frac{1}{4} \norm{h_{T^c}}_1.
\end{split}
\end{equation*}

\noindent Now let us denote
$$
f^+ := -\frac{1}{2} (A_{\Omega \bullet} g^{(x)} + g^{(e)}) \quad \text{and} \quad f^- := -\frac{1}{2} (A_{\Omega \bullet} g^{(x)} - g^{(e)}),
$$
as well as
$$h^+ := -A^*_{ \Omega \bullet} f^+  \quad \text{and} \quad h^- := -A^*_{ \Omega \bullet} f^-  +  A^*_{ \Omega^c \bullet } A_{\Omega^c \bullet} g^{(x)}.
$$

\noindent It is easy to establish the following properties from this construction
\begin{equation}
\begin{cases}
\label{eqt::properties of h+ h- f+ f-}
g^{(x)} = -h^+ + h^- \\
g^{(e)} = -f^+ + f^- \\
\norm{h^+}_2 = \norm{f^+}_2 \leq \sigma  \\
M_1 = \norm{A_{\Omega^c \bullet} g^{(x)}}^2_2 + 2 \norm{f^-_S}_2^2  \\
M_2 = 2 \norm{f^-_J}_2^2.
\end{cases}
\end{equation}

\subsubsection{Bound $M_2$}

At first, since $(x^{\star} + g^{(x)}, e^{\star} + g^{(e)})$ is the pair of optimal solution of the convex program, we have $\norm{x^{\star}}_1 + \norm{e^{\star}}_1 \geq \norm{x^{\star} + g^{(x)}}_1 + \norm{e^{\star} + g^{(e)}}_1$. Furthermore, decomposing $g^{(x)}$ and $g^{(e)}$ and using the triangular inequality, we can derive
\begin{equation}
\label{inq::bound l1 norm of x+gx plus e+ge}
\begin{split}
&\norm{x^{\star} + g^{(x)}}_1 + \lambda \norm{e^{\star} + g^{(e)}}_1 \\
&= \norm{x^{\star} - h^+ + h^-}_1 + \lambda \norm{e^{\star} - f^+ + f^-}_1 \\
&\geq \norm{x^{\star} + h^-}_1 + \lambda \norm{e^{\star} + f^-}_1  - ( \norm{h^+}_1 + \lambda \norm{f^+}_1 ).
\end{split}
\end{equation}

\noindent Applying Lemma \ref{lem::construction of dual certificate (z, e) - Sub orthogonal matrix - tradition} together with the observation that $f^- = - A_{\Omega \bullet} h^-$ yields
\begin{multline*}
\norm{x^{\star} + h^-}_1 + \lambda \norm{e^{\star} + f^-}_1  \\
\geq \norm{x^{\star}}_1 + \lambda \norm{e^{\star}}_1 + \frac{\lambda}{4} \norm{f^-_{J}}_1 + \frac{1}{4} \norm{h^-_{T^c}}.
\end{multline*}

\noindent  Combining these arguments, we get
$$
\frac{\lambda}{4} \norm{f^-_{J}}_1 + \frac{1}{4} \norm{h^-_{T^c}}_1 \leq \norm{h^+}_1 + \lambda \norm{f^+}_1.
$$

\noindent Converting both sides to the $\ell_2$-norm using the crude inequality $\norm{a}_2 \leq \norm{a}_1 \leq \sqrt{n} \norm{a}_2$ for all $a \in \R^n$, then applying $\norm{f^+}_2 = \norm{h^+}_2 \leq \sigma$, we obtain the bound
\begin{equation}
\begin{split}
\label{inq::bound l2 norm of f^-_J + h^-_Tc - last step}
\frac{\min \{\lambda, 1 \}}{4} (\norm{f^-_{J}}_2 + \norm{h^-_{T^c}}_2) &\leq \sqrt{n} (1 + \lambda) \norm{f^+}_2  \\
&\leq \sqrt{n} (1 + \lambda) \sigma .
\end{split}
\end{equation}

\noindent A specific consequence of this analysis is a bound of $M_2$
\begin{equation}
\label{inq::bound M2}
\begin{split}
M_2 = 2 \norm{f^-_J}_2^2 &\leq 2 (\norm{f^-_{J}}_2 + \norm{h^-_{T^c}}_2)^2  \\
&\leq 2 \left( \frac{4 (1 + \lambda ) }{\min \{\lambda, 1 \}} \right)^2 \sigma^2 n.
\end{split}
\end{equation}

\subsubsection{Bound $M_1$}

In this section we would like to bound $\norm{A_{\Omega^c \bullet} g^{(x)}}_2^2 + 2 \norm{f^-_{S}}_2^2$. Denoting $z := \left(
  \begin{array}{c}
    f^-_S \\
    -A_{\Omega^c \bullet} g^{(x)} \\
  \end{array}
\right)$, then to bound the quantity of interest, it is equivalent to bounding $\norm{z}_2$. By the construction of $h^-$ and $f^-$, we have $A^*_{\Omega \bullet } f^- + h^- - A^*_{\Omega^c \bullet } A_{\Omega^c \bullet} g^{(x)}= 0$ leading to
\begin{equation}
\label{eqt::main equation to bound M1}
\begin{split}
-A^*_{ J \bullet} f^-_J - \left(
  \begin{array}{c}
    0_T \\
    h^-_{T^c} \\
  \end{array}
\right)
&= A^*_{ S \bullet} f^-_S  - A^*_{ \Omega^c \bullet} A_{\Omega^c \bullet} g^{(x)} + \left(
  \begin{array}{c}
    h^-_T \\
    0_{T^c} \\
  \end{array}
\right) \\
&= A^*_{J^c \bullet}  \left(
  \begin{array}{c}
    f^-_S \\
    -A_{\Omega^c \bullet} g^{(x)} \\
  \end{array}
\right)  +  \left(
  \begin{array}{c}
    h^-_T \\
    0_{T^c} \\
  \end{array}
\right) \\
&= A^*_{J^c \bullet}  z  +  \left(
  \begin{array}{c}
    h^-_T \\
    0_{T^c} \\
  \end{array}
\right),
\end{split}
\end{equation}
where the second identity follows from $J^c = S \cup \Omega^c$.

First we control the upper bound of the $\ell_2$-norm of the left-hand side of (\ref{eqt::main equation to bound M1}), which can be attained easily from the triangular inequality
\begin{align*}
\norm{A^*_{J \bullet} f^-_J + \left(
  \begin{array}{c}
    0_T \\
    h^-_{T^c} \\
  \end{array}
  \right)}_2 &\leq \norm{A^*_{J \bullet } f^-_J}_2 + \norm{h^-_{T^c}}_2 \\
&= \norm{f^-_J}_2 + \norm{h^-_{T^c}}_2.
\end{align*}

Next, the $\ell_2$-norm of the right-hand side of (\ref{eqt::main equation to bound M1}) is now lower bounded by
\begin{align*}
&\norm{A^*_{J^c \bullet } z + \left(
  \begin{array}{c}
    h^-_T \\
    0_{T^c} \\
  \end{array}
\right)}^2_2 \\
&= \norm{A^*_{J^c \bullet } z}_2^2 + \norm{h^-_T}^2_2 + 2 \inner{A^*_{ J^c T} z, h^-_T} \\
&\geq \norm{z}_2^2 + \norm{h^-_T}^2_2 - 2 \norm{A^*_{ J^c T}} \norm{z}_2 \norm{h^-_T}_2 \\
&\geq \norm{z}_2^2 + \norm{h^-_T}^2_2 - 2 \norm{A^*_{J^c T}} \norm{z}_2 \norm{h^-_T}_2 \\
&\geq \norm{z}_2^2 + \norm{h^-_T}^2_2 - 2 \sqrt{1 - \rho_0/2} \norm{z}_2 \norm{h^-_T}_2 \\
&\geq (1 - \sqrt{1 - \rho_0 / 2}) (\norm{z}_2^2 + \norm{h^-_T}^2_2),
\end{align*}
where the third inequality follows from Proposition \ref{lem::bound spectral norm of A_Jc Tc}: $\norm{A^*_{J^c T} A_{J^c T}} \leq 1 - \rho_0/2$ and the last inequality follows from the standard argument $a^2 + b^2 - 2\alpha ab \geq (1-\alpha)(a^2 + b^2)$.

Combine these pieces together with the fact that $1 - \sqrt{1 - \rho_0 / 2} \geq \frac{\rho_0}{4}$, we attain
\begin{equation}
\label{inq::bound l2 norm of z + h^-_T - last step}
\norm{z}_2^2 + \norm{h^-_T}_2^2 \leq \frac{4}{\rho_0 }  (\norm{f^-_J}_2 + \norm{h^-_{T^c}}_2)^2.
\end{equation}

\noindent Next, notice that $\norm{z}_2^2 = \norm{f^-_S}_2^2 + \norm{A_{\Omega^c \bullet} g^{(x)}}_2^2$ and together with (\ref{inq::bound M2}), we get the following bound of $M_1$
\begin{equation}
\label{inq::bound M1}
\begin{split}
M_1 \leq 2 (\norm{z}_2^2 + \norm{h^-_T}_2^2) &\leq \frac{8}{\rho_0 } (\norm{f^-_J}_2 + \norm{h^-_{T^c}}_2)^2 \\
&\leq \frac{8}{\rho_0} \left( \frac{4 (1 + \lambda)}{\min \{1, \lambda \}} \right)^2 \sigma^2 n.
\end{split}
\end{equation}

Obviously, from combining these two previous inequalities on $M_1$ and $M_2$, we can establish the bound of the sum $M_1 + M_2$. However, we can tighten this bound by a constant factor from the following simple steps:
\begin{align*}
M_1 + M_2 &\leq 2 (\norm{A_{\Omega^c \bullet} g^{(x)}}_2^2 + \norm{f^-_S}_2^2 + \norm{f^-_J}_2^2) \\
&\leq 2 \left[ \frac{4}{\rho_0} (\norm{f^-_J}_2 + \norm{h^-_{T^c}}_2)^2 + \norm{f^-_J}_2^2 \right] \\
&\leq 2 \left( \frac{4}{\rho_0} + 1 \right) (\norm{f^-_J}_2 + \norm{h^-_{T^c}}_2)^2 \\
&\leq 2 \left( \frac{4}{\rho_0} + 1 \right) \left( \frac{4 (1 + \lambda ) }{\min \{\lambda, 1 \}} \right)^2 \sigma^2 n,
\end{align*}
where the second inequality follows from (\ref{inq::bound l2 norm of z + h^-_T - last step}) and the last inequality follows from (\ref{inq::bound l2 norm of f^-_J + h^-_Tc - last step}).

Inserting the above bound into (\ref{inq::bound l-2 of g^x + l-2 of g^2 - first step}) leads to
$$
\norm{g^{(x)}}_2^2 + \norm{g^{(e)}}_2^2 \leq 2 \sigma^2 + 2 \left( \frac{4}{\rho_0}+1 \right) \left( \frac{4(\lambda+1)}{\min \{1, \lambda \}}  \right)^2  \sigma^2 n.
$$

\noindent Finally, applying $ (\norm{g^{(x)}}_2 + \norm{g^{(e)}}_2)^2 \leq 2 (\norm{g^{(x)}}_2^2 + \norm{g^{(e)}}_2^2 )$ will complete our proof.
\end{proof}

\section{Oracle inequalities}
\label{sec::oracle inequalities}

In this section we would like to discuss the optimality of the reconstruction error bound in Theorem \ref{thm::main theorem - with dense noise and sparse noise}. In particular, we compare this result with the best possible accuracy one can achieve. Suppose we had available an oracle informing us in advance the locations of $T$ nonzero coefficients of the signal as well as $S$ nonzero coefficients of the sparse noise. Then one can use this valuable information to construct the ideal estimator pair $(x^{\Oracle}, e^{\Oracle})$ by least-square projection. To see this, we decompose $y$ into two components: $y_S$ and $y_J$, where $y_J$ is not affected by sparse error. Thus,
$$
y_J = A_{JT}x^{\star}_T + \nu_J.
$$

\noindent Recall from (\ref{inq::bound spectral norm of I - rho A*_JT A_IT}), $A^*_{JT} A_{JT}$ is invertible. In particular, $\rho_0/2 \leq \sigma_{\min} (A^*_{JT} A_{J T}) \leq \sigma_{\max} (A^*_{JT} A_{J T}) \leq 3 \rho_0/2$ where $\sigma_{\min}$ and $\sigma_{\max}$ are the minimum and the maximum singular value of the matrix, respectively. Therefore, the least-square solution of this linear system is
$$
x^{\Oracle}_T = (A_{JT}^* A_{JT})^{-1} A_{JT}^* y_J.
$$

The oracle error bound on the signal is now estimated by
$$
\norm{x^{\Oracle}_T - x^{\star}_T}_2 = \norm{(A_{JT}^* A_{JT})^{-1} A_{JT}^* \nu_J}_2.
$$

\noindent It is obvious that $\norm{H^{-1}} \leq \frac{1}{\sigma_{\min} (H)}$ for any matrix $H$. Therefore, \begin{equation}
\norm{x^{\Oracle}_T - x^{\star}_T}_2 \leq \norm{(A_{JT}^* A_{JT})^{-1}} \norm{A_{JT}} \norm{\nu_J}_2 \leq \sigma \sqrt{6/\rho_0}.
\end{equation}

Now the oracle solution of the error can be found from the identity equation on the set $S$: $y_S = A_{ST} x^{\Oracle}_T + e^{\Oracle}_S + \nu_S$. This leads to
$$
e^{\Oracle}_S = e^{\star}_S + A_{ST} (x^{\star}_T - x^{\Oracle}_T).
$$

\noindent Recall in Proposition \ref{prop::spectral norm bound of A_TS} that
\begin{align*}
\norm{A_{ST}} &\leq (\eta \rho)^{1/2} \left( 1 + \sqrt{\frac{\mu k \log n}{s}} \right)^{1/2} \\
&= \eta^{1/2} \left( \frac{s}{m} + \sqrt{\frac{s}{m}} \sqrt{\frac{\mu k \log n}{m}}  \right)^{1/2} \leq \sqrt{2},
\end{align*}
provided that $m \geq \mu k \log n$. We conclude that the oracle error bound on $e^{\star}$ has to satisfy
$$
\norm{e^{\Oracle}_S - e^{\star}_S}_2 \leq \sqrt{2} \norm{x^{\Oracle}_T - x^{\star}_T}_2 \leq \sqrt{12/\rho_0}.
$$

\noindent In conclusion, with the help of the oracle, we have
\begin{equation}
\norm{x^{\Oracle} - x^{\star} }_2 + \norm{e^{\Oracle} - e^{\star} }_2 \leq 3 \sigma \sqrt{6/\rho_0}.
\end{equation}
with adversarial noise. Consequently, our error bound in Theorem \ref{thm::main theorem - with dense noise and sparse noise} loses a $\sqrt{n}$ vis-a-vis over the ideal bound achieved via the oracle help.

\section{Numerical experiments}
\label{sec::experiments}

In this section, we provide extensive simulations to confirm the validity of our theoretical results. Since the observation model in (\ref{eqt::measurement model}) can be expressed as $y = [A_{\Omega \bullet} \quad \frac{1}{\lambda} I] z^{\star} = B z^{\star}$ where $z^{\star} = [x^{\star^T}, \lambda e^{\star^T}]^T$ and $I$ is the $m \times m$ identity matrix, the extended $\ell_1$-minimization in (\ref{opt::l1-l1 minimization - sub orthogonal matrix}) and the noisy version in (\ref{opt::l1-l1 minimization with noise - sub orthogonal matrix}) can be recast as conventional $\ell_1$ programs
$$
\min_z \norm{z}_1 \quad \text{s.t.}  \quad y = B z,
$$
and
$$
\min_z \norm{z}_1 \quad \text{s.t.}  \quad \norm{b - B z} \leq \sigma.
$$
In this section, we use the Homotopy solver introduced in \cite{AR_homotopy_2011_J} for our experiments. Another important implementation detail is the choice of the parameter $\lambda$. For moderate signal dimensions (e.g $n \leq 10^8$), we suggest to set $\lambda = \sqrt{\frac{n}{m (\log n)^{1/2}}}$. With this choice, measurements are allowed to be corrupted up to $25\%$ as presented in our theorems. Of course, if we know in prior that the signal is very sparse, reducing the value of $\lambda$ will help retrieve the signal under more corrupted measurements. In practical applications, we recommend $\lambda = \sqrt{\frac{n}{m (\log n)^{1/2}}}$ as a "good-for-all" parameter.


\subsection{Exact recovery from grossly corrupted measurements}

We first illustrate the correct recoverability of the signal under gross error as provided in Theorem \ref{thm::detail of the main theorem}. We consider random signals $x^{\star}$ of varying lengths $n = \{1024, 2048, 4096, 8192 \}$. For each $n$, we generate signals of sparsity $k$ where $k$ varies from $1$ to $60$ with step size $2$. Here, magnitudes of nonzero entries are Gaussian distributed and their locations are chosen uniformly at random. For each sparsity level, the measurement matrix $A_{\Omega \bullet}$ is produced by uniformly selecting $m=500$ rows at random from the Fourier matrix $A$. Error vector $e^{\star}$ is generated to have uniformly distributed support with cardinality $s = m/4$ and the polarity of nonzero entries being equally likely positive or negative. We set magnitudes of $e^{\star}$ such that $\norm{e^{\star}}_2 \geq 100 \norm{x^{\star}}_2$. The reader should note that this setting yields an observed signal that is significantly dominated by the noise.

For each value of the signal sparsity $k$, we repeat the experiment $100$ times and keep track of the probability of exact recovery. In all experiments, we set $\lambda = \sqrt{\frac{n}{m (\log n)^{1/2}}}$. The algorithm is declared to be successful if the relative error with respect to $x^{\star}$ satisfies $\norm{\widehat{x} - x^{\star}}_2/\norm{x^{\star}}_2 \leq 10^{-3}$. The performance curve is plotted in Fig. \ref{fig::noiseless case}. Numerical values on the x-axis denote signal sparsity whereas those on the y-axis denote the probability of exact recovery. Interestingly, this experiment demonstrates that the theory provides an accurate prediction of the simulation results even for relatively small problem sizes. In particular, perfect recovery is still attained with signals of moderate sparsity level even if $25\%$ measurements are grossly perturbed. Furthermore, the sparsity level is proportional with $\frac{m}{(\log n)^{3/2}}$ as expected.

Next, we fix the signal dimension to $n=1024$ and performs the same experiments with varying signal sparsity $k = [20, 25, 30]$. Fig \ref{fig::probability vs fraction error} demonstrates the probability of success with varying fraction error $s/m$. Note that as the signal's sparsity level increases, we expect to recover the signal under fewer corrupted measurements.

\begin{figure}[!t]
\centering
\includegraphics[width=2.4in]{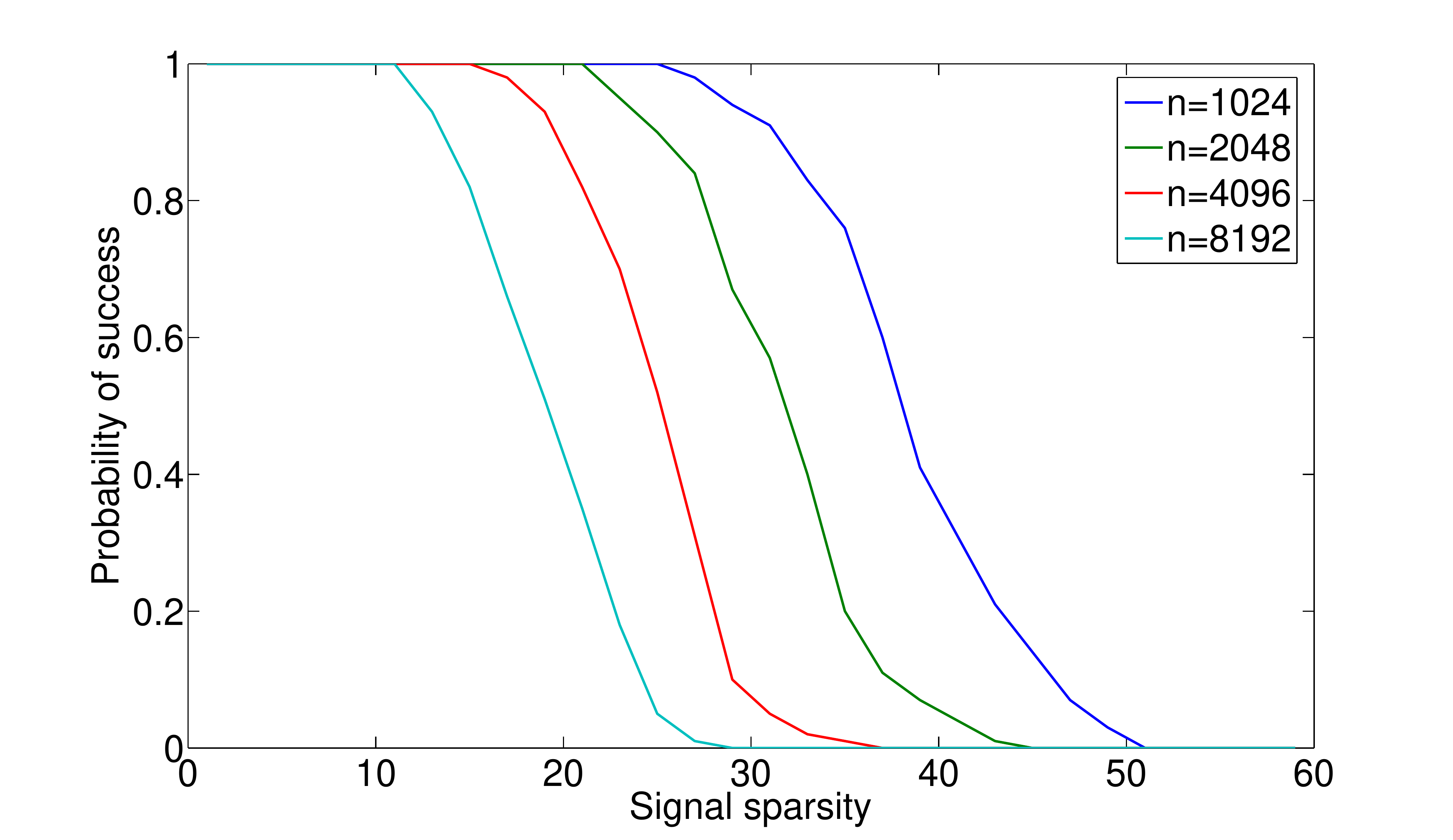}
\caption{The probability of success as a function of signal sparsity for various signal dimensions. Here, a total of $m=500$ measurements are observed and $1/4$ of them are grossly corrupted.}
\label{fig::noiseless case}
\end{figure}

\begin{figure}[!t]
\centering
\includegraphics[width=2.4in]{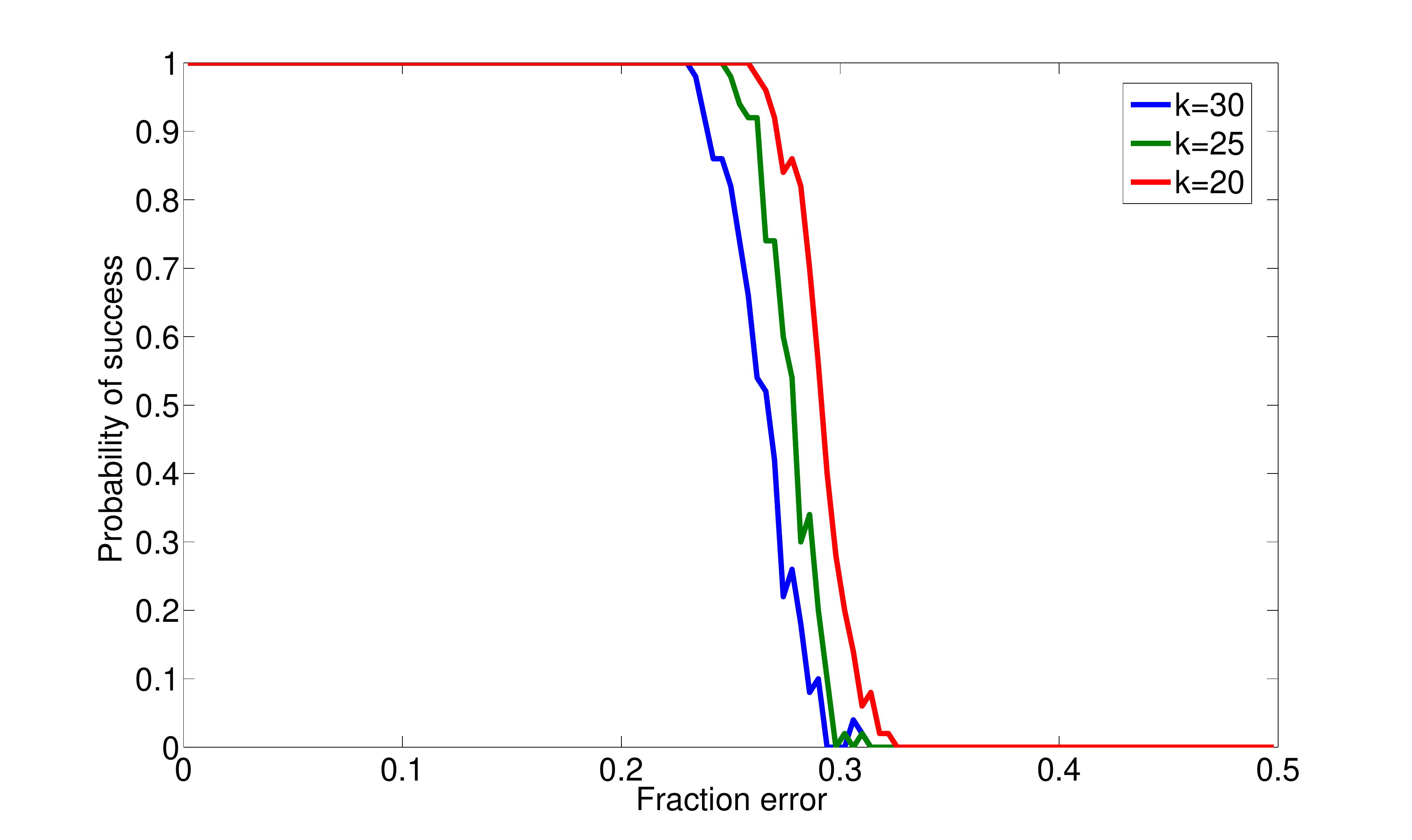}
\caption{The probability of success as a function of fraction error $s/m$. Here, we fix signal dimension to $n=1024$, a total of $m=500$ measurements are used and the signal sparsity is $k = [25, 30, 35]$.}
\label{fig::probability vs fraction error}
\end{figure}

\subsection{Stable recovery from both dense and sparse corrupted measurements}

We now demonstrate stable recoverability when measurements are both contaminated by gross sparse and small dense error. We generate small noise $\nu$ from i.i.d. $\oper N(0,\delta^2)$. The signal $x^{\star}$, the sparse error $e^{\star}$ and the measurement matrix $A_{\Omega\bullet}$ are constructed similarly as in previous experiments. For each setting, we perform the simulations $100$ times and report the average error.

We first evaluate the performance of (\ref{opt::l1-l1 minimization with noise - sub orthogonal matrix}) with the signal $x^{\star}$ whose dimension and sparsity level are fixed to be $n = 1024$ and $k = 20$. We also set the number of measurements and the error sparsity to be $m = 500$ and $s = m/4 $. Non-zero entries of the signal and the sparse errors are i.i.d. $\oper N(0,10)$. Estimation errors are quantified by the root-mean square (RMS), which is defined as $\norm{\widehat{x} - x^{\star}}_2/n$ and $\norm{\widehat{e}-e^{\star}}_2/n$, respectively. Fig. \ref{fig::RMS with sigma} shows the RMS error with varying noise level. We also demonstrate in this figure the RMS errors of an oracle obtained from Section \ref{sec::oracle inequalities}. Fig. \ref{fig::RMS with sigma} clearly illustrates that the RMS errors grow almost linearly with the noise level. Furthermore, RMS errors attained by solving (\ref{opt::l1-l1 minimization with noise - sub orthogonal matrix}) is just twice the RMS error achieved by the oracle.

Now we fix $\sigma = 1$ and run the optimization in (\ref{opt::l1-l1 minimization with noise - sub orthogonal matrix}) for varying values of error sparsity. Fig. \ref{fig::RMS with s} establishes fact that as $s$ decreases, we expect to achieve more accurate recovery.


\begin{figure}[!t]
\centering
\includegraphics[width=2.4in]{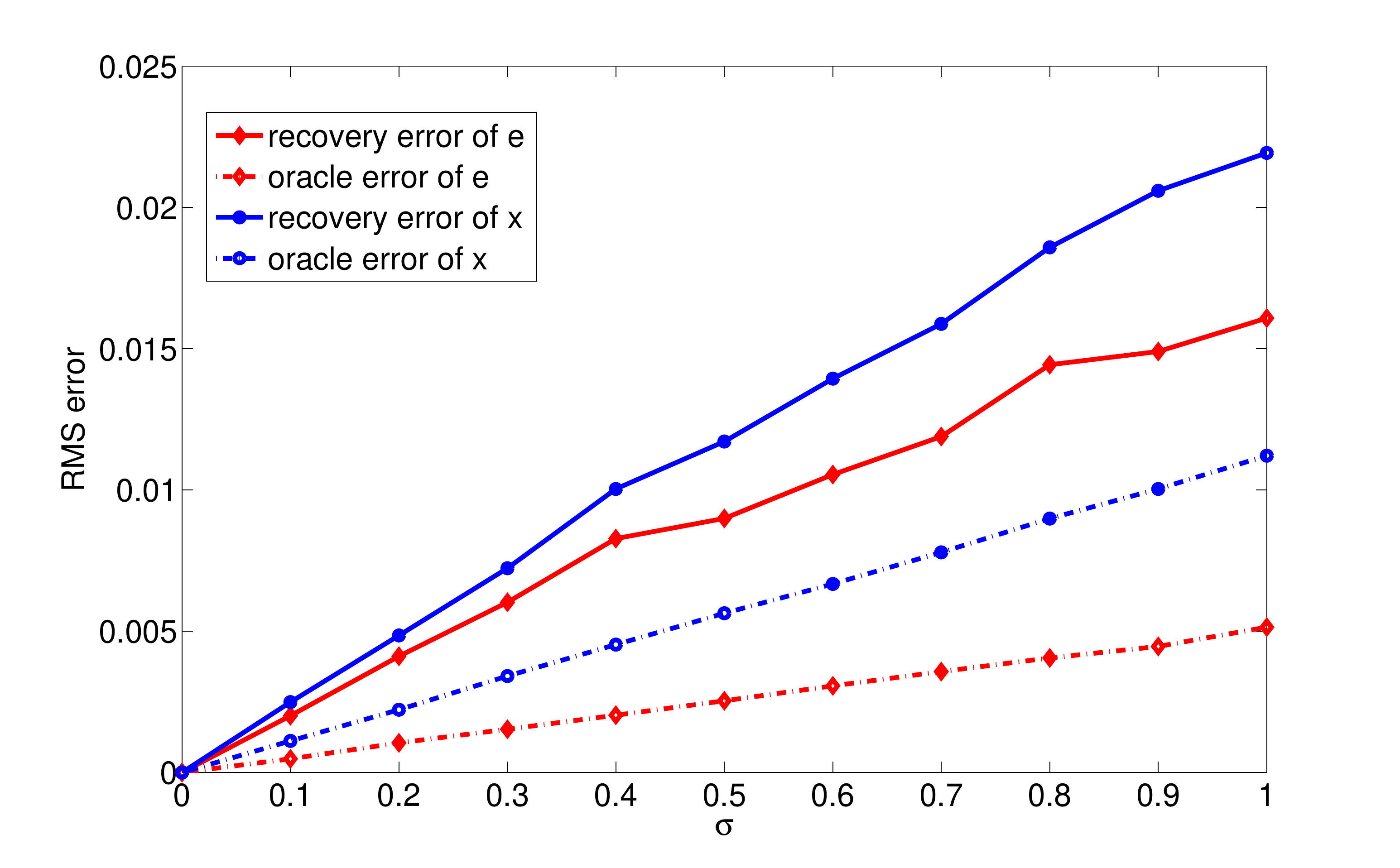}
\caption{RMS error as a function of $\sigma$ with $n=1024$, $m = 500$, $k=20$ and $s = m/4$.}
\label{fig::RMS with sigma}
\end{figure}

\begin{figure}[!t]
\centering
\includegraphics[width=2.4in]{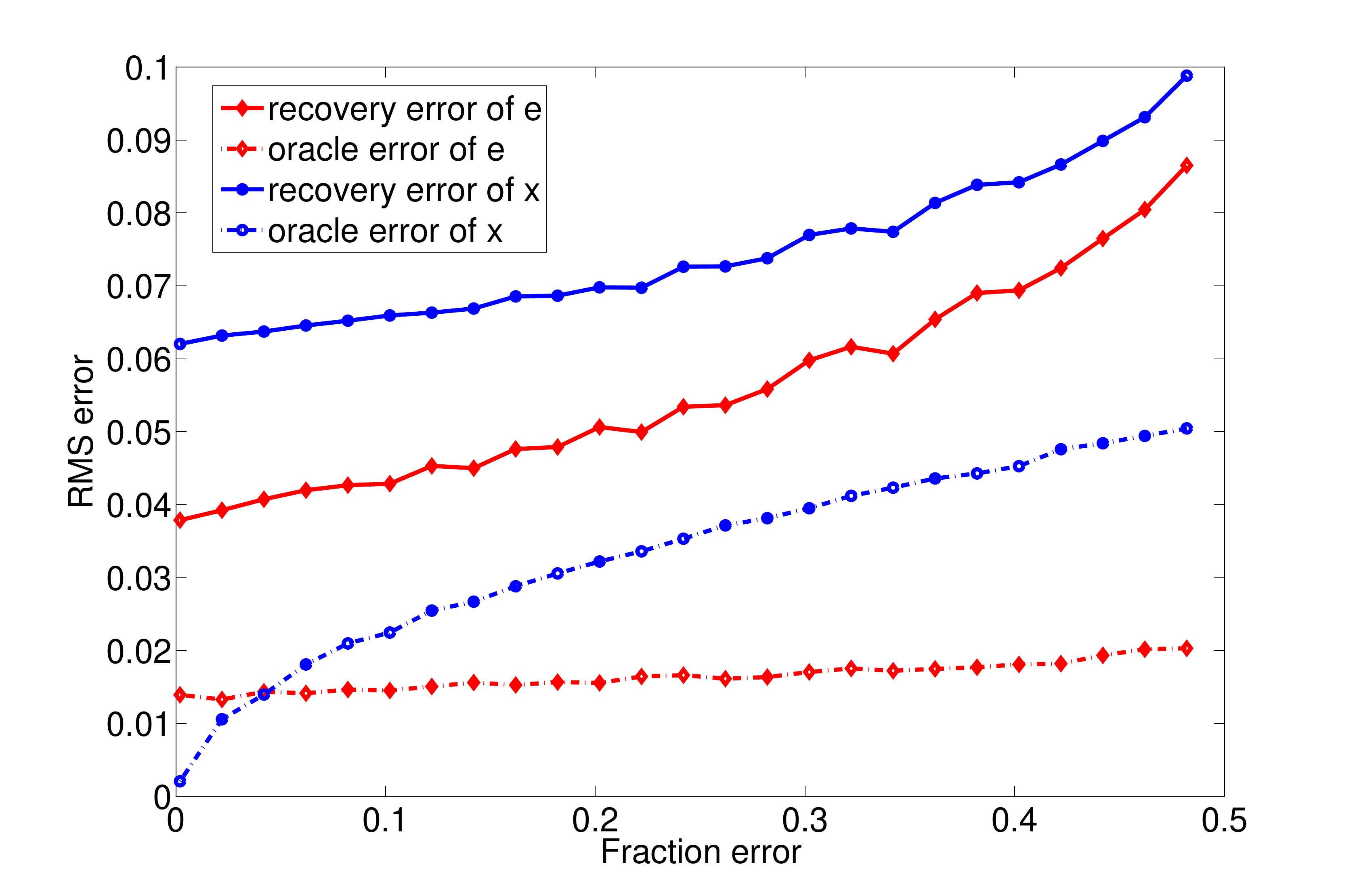}
\caption{RMS error as a function of $s$ with $n=1024$, $m = 500$, $k=20$ and $\sigma = 1$.}
\label{fig::RMS with s}
\end{figure}

\subsection{Experiments with images}

In our last experiment, we consider the problem of recovering an image from highly corrupted undersampled Fourier coefficients. As usual, the data is given by $y = A_{\Omega \bullet} x^{\star} + e^{\star} + \nu$ where $A_{\Omega \bullet}$ is a partial Fourier matrix obtained from subsampling rows of the full 2D Fourier matrix $A$, $e^{\star}$ is a sparse error vector whose nonzero entries can have arbitrarily large magnitudes, and $\nu$ is a small dense noise vector. In this experiment, $x^{\star}$ is the Shepp-Logan phantom image (see Fig. \ref{fig::phantom}), which is not sparse in the spatial domain but in the gradient domain. Therefore, to reconstruct $x^{\star}$, we use the total variation (TV) criterion and minimize
\begin{equation}
\label{opt::TV-L1 norm}
\min_{x,e} \norm{x}_{\TV} + \lambda \norm{e}_1 \quad \text{s.t.} \quad \norm{y - A_{\Omega \bullet}x - e}_2 \leq \sigma,
\end{equation}
where $\norm{\nu}_2 \leq \sigma$ is assumed to be known and $\norm{x}_{\TV}$ is the $\ell_1$-norm of the gradient, also known as the total-variation of $x$. This norm is formally defined as
\begin{equation}
\norm{x}_{\TV} = \sum_{ij} \sqrt{(\nabla_h x)^2_{ij} + (\nabla_v x)^2_{ij}},
\end{equation}
where $\nabla_h$ and $\nabla_v$ denote the discrete finite difference operators along the horizonal and vertical coordinates. To optimize (\ref{opt::TV-L1 norm}), we employ the classic alternating direction method (ADM) as presented in \cite{YZY_2010_J}. In this particular experiment, we perform a two-step algorithm
\begin{enumerate}
    \item We solve (\ref{opt::TV-L1 norm}) via the ADM method. The optimal solution is denoted as $(\widehat{x}, \widehat{e})$.
    \item Next, we select $J \in \{1,...,m \}$ as locations where coefficients of $\widehat{e}$ are zeros or approximately zeros. These locations correspond to reliable observations. Then, we solve the following optimization
    \begin{equation}
    \label{opt::TV opt}
    \min_{x} \norm{x}_{\TV}  \quad \text{s.t.} \quad \norm{y_J - A_{J \bullet}x}_2 \leq \sigma,
    \end{equation}
    where only clean observations are considered. The output of (\ref{opt::TV opt}) is what we expect to get.
\end{enumerate}


In this experiment, we sample $12267$ Fourier coefficients of the $256 \times 256$ phantom image $x^{\star}$ along a number of radical lines (as seen in the top right of Fig. \ref{fig::phantom}, $45$ radical lines are sampled). We then select $50 \%$ of these coefficients uniformly at random and purposely add them to a deterministic large noise vector whose magnitudes are twice larger than the magnitudes of Fourier coefficients. This process assumes that half of the observed Fourier coefficients are significantly corrupted during the data acquisition. We note that the locations of these missing entries are unknown. All the Fourier coefficients is afterward contaminated by a Gaussian noise vector with zero mean and standard deviation $0.01$. Fig. \ref{fig::phantom} on the bottom left and right shows the reconstruction from minimizing the TV only and from the aforementioned two-step algorithm, respectively. In the optimization (\ref{opt::TV-L1 norm}), $\lambda$ is set to be $\sqrt{\frac{n}{m \log n}}$. It is clear that while the conventional TV minimization fails to recover the original image, our proposed method recovers the image almost exactly. Notably, the relative error $\frac{\norm{x^{\star} - x_{\text{recovered}}}_2}{\norm{x^{\star}}_2}$ of our method is $0.0887$.


\begin{figure}[!t]
\centering
\subfigure{
\includegraphics[width=1.2in]{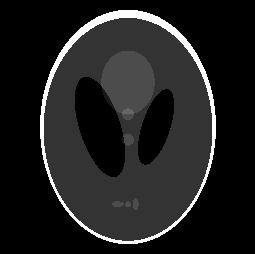} }
\subfigure{
\includegraphics[width=1.2in]{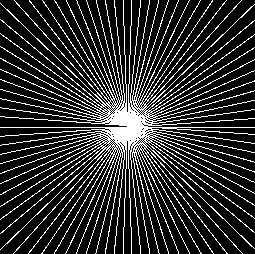} }
\subfigure{
\includegraphics[width=1.2in]{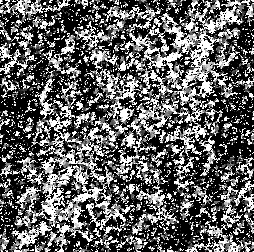} }
\subfigure{
\includegraphics[width=1.2in]{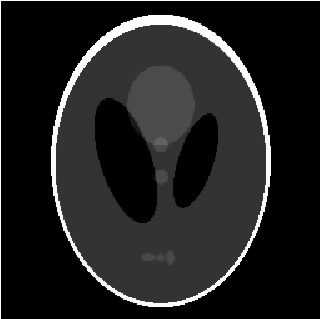} }
\caption{Top left: original $256 \times 256$ phantom image. Top right: Fourier domain sampling positions with $45$ radical lines. Bottom left: recovered image from the TV only. Bottom right: recovered image from our proposed optimization in (\ref{opt::TV-L1 norm}).}
\label{fig::phantom}
\end{figure}

\section{Discussion and conclusion}
\label{sec::conclusion}

In this paper, we present a complete analysis of a surprising phenomenon: one can recover perfectly a sparse signal from grossly corrupted measurements by linear programming (\ref{opt::l1-l1 minimization - sub orthogonal matrix}), even if the corruption is up to a significant fraction of all the entries. More specifically, we deliver an explicit connection between sparsity levels of the signal and the error. Our result can be interpreted as a generalization of compressed sensing, where measurements are both incomplete and corrupted by sparse errors. Furthermore, our results indicate that robustness is still retained even in a more challenging situation: the convex program (\ref{opt::l1-l1 minimization with noise - sub orthogonal matrix}) can stably recover a sparse signal under measurements perturbed by both gross sparse and small dense errors. Particularly, recovery error lies within a constant fraction of the dense noise level. We also establish stable recovery for a much more general class of signals $-$ approximately sparse signals.

As exhibited in Theorem \ref{thm::main theorem}, when the fraction of error is close to $1$ $-$ or in other words, most of the measurements are corrupted, signal sparsity $k$ is still allowed to be proportional to $\frac{m}{\mu^2 log^2 n}$ in order to retain accurate recovery. We conjecture that this bound is optimal. That is, we cannot achieve perfect reconstruction when $k \sim O(\frac{m}{\mu^2 \log n})$ and the error support size $s$ is close to $m$. In fact, we claim this conjecture in our upcoming paper for a class of Gaussian measurement matrices \cite{NNT_RobustLasso_2011_C}. How to establish a similar result for suborthogonal measurement matrices is an interesting open problem.


We would like to mention a related work that describes a similar phenomenon. Recently, Cand\`es \textit{et al.} \cite{CLMW_RobustPCA_2009_J}, \cite{GWLCM_RPCA_2010_C}, Chandrasekaran \textit{et al.} \cite{CSPW_2011_J}, Xu \textit{et al.} \cite{XCS_RPCA_2010_C} and Agarwal \textit{et al.} \cite{ANW_RPCA_2011_C} have shown that one can exactly recover a low-rank matrix $L \in \R^{n_1 \times n_2}$ from its grossly corrupted entries $M = L + S$ by solving the following convex program:
\begin{equation}
\label{opt::Robust PCA}
\min_{L, S} \norm{L}_* + \lambda \norm{S}_1 \quad\quad \text{subject to} \quad\quad M = L + S.
\end{equation}
More specifically, the authors of \cite{CLMW_RobustPCA_2009_J}, \cite{GWLCM_RPCA_2010_C} proved that as long as the rank of $L$ is an order of $\frac{n}{\log^2 n}$ with $n = \max\{n_1, n_2\}$, then the solution of (\ref{opt::Robust PCA}) with an appropriate choice of parameter $\lambda$ is exact even if almost all entries of $L$ are arbitrarily perturbed. Interestingly, the results in these papers shares similar behavior as what presented here in our paper. We believe that similar phenomena also holds for other high-dimensional signal and error models as well.

\section{Appendix}

\begin{proof} [\textit{Proof of Corollary \ref{corr::main corr - with generic signal and dense noise and sparse noise}}]

At first, we observe a variant of Lemma \ref{lem::construction of dual certificate (z, e) - Sub orthogonal matrix - tradition}. Assuming the existence of a dual vector ($z^{(x)}, z^{(e)}$) satisfying properties of Lemma \ref{lem::construction of dual certificate (z, e) - Sub orthogonal matrix - tradition}, then for any perturbation pair ($f, h$) such that $f = - A_{\Omega \bullet} h$, we have
\begin{equation}
\label{inq::main lemma with non-sparse signal}
\begin{split}
\norm{x^{\star} + h}_1 &+ \lambda \norm{e^{\star} + f}_1 \geq \norm{x^{\star}_T}_1 + \lambda \norm{e^{\star}}_1 \\
&{ }- \norm{x^{\star}_{T^c}}_1 + \frac{1}{4} (\norm{h_{T^c}}_1 + \lambda \norm{A_{J \bullet} h}_1).
\end{split}
\end{equation}

\noindent The proof is essentially analogous to that of Lemma \ref{lem::construction of dual certificate (z, e) - Sub orthogonal matrix - tradition}. The only difference is the non-sparse nature of $x^{\star}$. Now decompose $x^{\star}$ into $x^{\star}_T$ and $x^{\star}_{T^c}$ and use the triangular inequality to provide a lower bound for $\norm{x^{\star} + h}_1$, we have
$$
\norm{x^{\star} + h}_1 + \lambda \norm{e^{\star} + f}_1 \geq \norm{x^{\star}_T + h}_1 + \lambda \norm{e^{\star} + f}_1 - \norm{x^{\star}_{T^c}}_1.
$$

\noindent Applying Lemma \ref{lem::construction of dual certificate (z, e) - Sub orthogonal matrix - tradition} to the bound $\norm{x^{\star}_T + h}_1 + \lambda \norm{e^{\star} + f}_1$ will lead to the inequality (\ref{inq::main lemma with non-sparse signal}).

Following closely the proof of Theorem \ref{thm::main theorem - with dense noise and sparse noise}, except in bounding the quantity $M_2$, we employ the inequality in (\ref{inq::main lemma with non-sparse signal}). With the same notations, we have $\norm{x^{\star} + g^{(x)}}_1 + \norm{e^{\star} + g^{(e)}}_1 \leq \norm{x^{\star}}_1 + \norm{e^{\star}}_1 = \norm{x^{\star}_T}_1 + \norm{x^{\star}_{T^c}}_1 + \norm{e^{\star}}_1$. Using the lower bound of $\norm{x^{\star} + g^{(x)}}_1 + \norm{e^{\star} + g^{(e)}}_1$ in (\ref{inq::bound l1 norm of x+gx plus e+ge}) together with (\ref{inq::main lemma with non-sparse signal}), we get a similar result as in (\ref{inq::bound l2 norm of f^-_J + h^-_Tc - last step})
$$
\frac{\min \{\lambda, 1 \}}{4} (\norm{f^-_{J}}_2 + \norm{h^-_{T^c}}_2) \leq \sqrt{n} (1 + \lambda) \sigma + 2 \norm{x^{\star}_{T^c}}_1 .
$$

\noindent The rest of our proof follows exactly from the analysis of Theorem \ref{thm::main theorem - with dense noise and sparse noise}.
\end{proof}

\begin{proof} [\textit{Proof of Lemma \ref{lem::show the almost orthogonality of A_ST^c versus A_ST}}]
The proof is essentially analogous to the one presented in \cite{CR_incoherence_2007_J}. We first establish a bound for $\E \norm{ A^*_{S_0 T} u }_{2}$, and then show that $\norm{ A^*_{S_0 T} u }_{2}$ concentrates around its expectation.

Define $S_0 = \{i: \delta_i = 1 \}$ where $\delta_i$ is an independent sequence of Bernoulli variables with parameter $\rho_0$ and denote by $v_i \in R^k$ the $i^{th}$ column of matrix $A^*_{ S_0 T}$. With these notations, we have
$$
A^*_{S_0 T} u = \sum_{i \in S_0} u_i v_i = \sum_{i=1}^n \delta_i u_i v_i.
$$

\noindent Notice that from the orthogonality property of $A$, $\sum_{i=1}^n u_i v_i = A^*_{\bullet T} a = 0$ where $a$ is a column of matrix $A_{\bullet T^c}$. Thus, by subtracting this zero term from $A^*_{S_0 T} u$, one can see that $A^*_{S_0 T} u$ is a sum of zero-mean random variable
$$
A^*_{S_0 T} u = \sum_{i=1}^n (\delta_i - \rho_0) u_i v_i.
$$

\noindent We can now estimate $\E \norm{A^*_{S_0 T} u}_2$ as follows
\begin{align*}
\E \norm{A^*_{S_0 T} u}_2^2 &= \E \sum_{i=1}^n (\delta_i - \rho_0)^2 u^2_i \inner{v_i, v_i} \\
&{ }+ \E \sum_{i,j; i\neq j} (\delta_i - \rho_0) (\delta_j - \rho_0) u_i u_j \inner{v_i, v_j}.
\end{align*}

\noindent The second term vanishes due to the independence of $\delta_i$, $i = 1,...,n$. Furthermore, $\E (\delta_i - \rho_0)^2 = \rho_0 (1 - \rho_0) \leq \rho_0$. Hence,
\begin{align*}
\E \norm{A^*_{S_0 T} u}_2^2  &\leq \rho_0 \max_i \norm{v_i}_2^2 (\sum_{i=1}^n u_i^2) \\
&= \rho_0 \max_i \norm{v_i}_2^2 \leq \rho_0 \frac{\mu k}{n}.
\end{align*}

\noindent Therefore, by Jensen's inequality, we conclude that $\E \norm{A^*_{S_0 T} u}_2 \leq \sqrt{\E \norm{A^*_{S_0 T} u}_2^2 } \leq \sqrt{\rho_0 \frac{\mu k}{n}}$.

We now apply a remarkable result from Talagrand that bounds the supremum of a sum of independent random variables. Let $Z_1,..., Z_n$ be a sequence of independent random variables and let $M$ be the supremum defined by
$$
M = \sup_{g \in \oper G} \sum_{i=1}^n g(Z_i),
$$
where $g$ is a family of real-valued functions.
\begin{thm}
\label{thm::Talagrand inequality}
If $|g| \leq B$ for every $g \in \oper G$ and $\{ g(Z_i) \}_{i=1,...,n}$ have zero mean for every $g \in \oper G$, then for all $\tau \geq 0$,
$$
\Prob \left( |M - \E M| \geq \tau \right) \leq 3 \exp \left( - \frac{t}{C_T B} \log \left( 1 + \frac{B \tau}{\sigma^2 + B \E \overline{M}} \right) \right),
$$
where $\sigma^2 = \sup_{g \in \oper G} \sum_{i=1}^n \E g^2 (Z_i)$, $\overline{M} = \sup_{g \in \oper G} | \sum_{i=1}^n g(Z_i) |$ and $C_T > 0$ is a small numerical constant.
\end{thm}

By the definition of norm, we have
\begin{align*}
M &:= \norm{A^*_{S_0 T} u}_2 = \max_{\norm{g}_2 \leq 1} \inner{A^*_{S_0 T} u, g} \\
&= \max_{\norm{g}_2 \leq 1} \sum_{i=1}^n (\delta_i - \rho_0) u_i \inner{v_i, g}.
\end{align*}

Denote $Z_i =  (\delta_i - \rho_0) u_i v_i$, we have $M$ is the supremum sum of independent random variable $g(Z_i)$ where $g(Z_i) := (\delta_i - \rho_0) u_i \inner{v_i, g}$. Since $M \geq 0$, $\E M = \E \overline{M}$. The absolute value of $g(Z_i)$ is bounded by
$$
|g(Z_i)|  \leq \norm{ (\delta_i - \rho_0) u_i v_i }_2 \norm{g}_2 \leq |u_i| \norm{v_i}_2 \leq \frac{\mu}{n} \sqrt{k} := B.
$$

\noindent In addition, from $\E (\delta_i - \rho_0)^2 = \rho_0 (1 - \rho_0)$, $\sigma^2$ is computed from the argument
\begin{align*}
\sum_{i=1}^n \E g^2 (Z_i) &= \sum_{i=1}^n \rho_0 (1 - \rho_0) u_i^2 \inner{v_i, g}^2 \\
& \leq \rho_0 \max_i u_i^2 g (\sum_{i=1}^n v_i v_i^*) g \leq \rho_0 \frac{\mu}{n} \norm{g}_2^2 \norm{\sum_{i=1}^n v_i v_i^*}.
\end{align*}

\noindent Notice that $\sum_{i=1}^n v_i v_i^* = A^*_{T \bullet} A_{\bullet T} = I$ by the orthogonality property of $A$. Then, $\sigma^2 \leq \max_{\norm{g}_2 \leq 1} \rho_0 \frac{\mu}{n} \norm{g}_2^2 \leq \rho_0 \frac{\mu}{n}$. Applying Talagrand's inequality yields
\begin{equation}
\label{inq::bound S by Talagrand inequality}
\begin{split}
&\Prob \left( M \geq \E M + \tau \right) \\
&\leq 3 \exp \left( - \frac{\tau}{C_T \sqrt{k} \mu/n} \log \left( 1 + \frac{ \tau \sqrt{k}}{\rho_0 + k \sqrt{\rho_0} (\mu/ n)^{1/2} } \right) \right).
\end{split}
\end{equation}

\noindent We need to consider two cases
\begin{enumerate}
	\item If $\rho_0 \geq k \sqrt{\rho_0 \frac{\mu}{n}}$, or equivalently, $\rho_0 \geq \frac{\mu k^2}{n}$, we select $\tau$ such that $\tau \leq \rho_0/\sqrt{k}$. Thus, the right-hand side of (\ref{inq::bound S by Talagrand inequality}) is bounded by
$$
3 \exp \left( - \frac{\tau}{C_T \sqrt{k} \mu/n} \log \left( 1 + \frac{\tau \sqrt{k}}{2 \rho_0} \right) \right) ,
$$
which is in turn smaller than $3 \exp \left( - \frac{\tau^2}{3C_T \rho_0 \mu / n} \right)$ due to the simple observation that $\log (1 + x) \geq 2x/3$ for $0 \leq x \leq 1$. Set $\tau^2 := C \rho_0 \frac{\mu \log n}{n}$ where $C = 15 C_T$, the right-hand side of (\ref{inq::bound S by Talagrand inequality}) will be less than $3 e^{-\log n^5} = 3 n^{-1}$. Note that this choice of $\tau$ is consistent with the condition $\tau \leq \rho_0/\sqrt{k}$ as long as $\rho_0 \geq C \frac{\mu k \log n}{n}$. We conclude that in this case
$$
\Prob \left( M \geq \sqrt{\rho_0 \frac{\mu k}{n} } + \sqrt{C \rho_0 \frac{\mu \log n}{n}} \right) \leq 3 n^{-1}.
$$

\noindent In other words, with high probability, $M \leq \sqrt{C' \rho_0 \frac{\max \{k,\log n \}}{n}}$.

	\item On the other hand, if $\rho_0 \leq \frac{\mu k^2}{n}$, we select $\tau$ such that $\tau \leq \sqrt{\rho_0 \frac{\mu k}{n}}$. The right-hand side of (\ref{inq::bound S by Talagrand inequality}) is now less than
\begin{align*}
&3 \exp \left( - \frac{\tau}{C \sqrt{k} \mu/n} \log \left( 1 + \frac{\tau }{2 (\rho_0 k \mu /n)^{1/2}} \right) \right) \\
&\leq 3 \exp \left( - \frac{\tau^2}{3C_T k \rho^{1/2}_0 (\mu / n)^{3/2}} \right).
\end{align*}

\noindent Similarly, the right-hand side of (\ref{inq::bound S by Talagrand inequality}) will be less than $3 n^{-1}$ by setting $\tau^2 := C k \rho_0^{1/2} (\frac{\mu}{n})^{3/2} \log n$. This choice of $\tau$ is consistent with its bound as long as $\rho_0 \geq C \frac{\mu (\log n)^2}{n}$. Therefore,
$$
\Prob \left( M \geq \sqrt{\rho_0 \frac{\mu k}{n} } + \sqrt{C k \log n} \rho_0^{1/4} (\frac{\mu}{n})^{3/4}  \right) \leq 3 n^{-1}.
$$

\noindent In other words, with high probability, $M \leq  \sqrt{C' \rho_0 \frac{\mu k}{n} }$ and the proof is completed.
\end{enumerate}
\end{proof}

\bibliographystyle{IEEEtran}
\bibliography{all_references}

\end{document}